\documentclass[11pt]{article}
\usepackage{amsmath,amsfonts,amsthm,amssymb,bbm,mathrsfs,enumerate,nicefrac}
\usepackage{graphicx, xcolor}
    \usepackage[driverfallback=hypertex,pagebackref=true,colorlinks]{hyperref}
    \hypersetup{linkcolor=[rgb]{.7,0,0}}
    \hypersetup{citecolor=[rgb]{0,.7,0}}
    \hypersetup{urlcolor=[rgb]{.7,0,.7}}
\usepackage{geometry}
\usepackage{epstopdf}
\usepackage[ruled]{algorithm2e}
\AtBeginDocument{%
  \addtolength\abovedisplayskip{-0.25\baselineskip}%
  \addtolength\belowdisplayskip{-0.15\baselineskip}%
}

\usepackage[titletoc,title]{appendix}
\usepackage[capitalise]{cleveref}
\usepackage{url}

\makeatletter
\newcommand{\mt}{\@ifnextchar^{}{^{}}}

\newcommand{\mc}{\mathcal}
\newcommand{\E}{\mathop{\mathbf{E}\mbox{}}\limits}

\renewcommand{\S}{\mathbf{S}}
\newcommand{\I}[3][]{\mathbf{I}_{#1}\left(#2;#3\right)}
\newcommand{\ct}{^\dagger}
\newcommand{\id}{\mathbb{I}\mt}

\newcommand{\midd}{\mathrel{\|}}

\newcommand{\ket}[1]{|#1\rangle}
\newcommand{\bra}[1]{\langle#1|}
\newcommand{\braket}[2]{\langle #1|#2\rangle}

\newcommand{\Tr}{\mathrm{Tr}}
\DeclareMathOperator{\diag}{diag}
\DeclareMathOperator{\Diag}{Diag}
\newcommand{\Var}{\mathrm{Var}}

\newcommand{\lN}[2][]{\big\|#2\big\|^{#1}}

\newcommand{\rhot}{\rho^{(t)}}
\newcommand{\taut}{\tau^{(t)}}

\newtheorem{theorem}{Theorem}
\newtheorem{lemma}{Lemma}[section]
\newtheorem{proposition}[lemma]{Proposition}

\newtheorem{corollary}[lemma]{Corollary}
\newtheorem{claim}[lemma]{Claim}
\theoremstyle{remark}

\theoremstyle{definition}
\newtheorem*{definition}{Definition}
\newtheorem*{example}{Example}

\title{Memory-Sample Lower Bounds for Learning with Classical-Quantum Hybrid Memory}
\author{Qipeng Liu \thanks{Simons Institute for Theory of Computing. E-mail: qipengliu0@gmail.com. Research supported in part by the Simons Institute for the Theory of Computing, through a Quantum Postdoctoral Fellowship, by the DARPA SIEVE-VESPA grant No.HR00112020023 and by the NSF QLCI program through grant number OMA-2016245.}% Any opinions, findings and conclusions or recommendations expressed in this material are those of the author(s) and do not necessarily reflect the views of the United States Government or DARPA.}
\and Ran Raz \thanks{Princeton University. E-mail: ranr@cs.princeton.edu. Research supported by a Simons Investigator Award and by the National Science Foundation grants No. CCF-1714779, CCF-2007462.}
\and Wei Zhan \thanks{Princeton University. E-mail: weizhan@cs.princeton.edu. Research supported by a Simons Investigator Award and by the National Science Foundation grants No. CCF-1714779, CCF-2007462.}}
\date{}

\begin{document}
\maketitle 

\begin{abstract}
In a work by Raz (J. ACM and FOCS 16), it was proved that any algorithm for parity learning on $n$ bits requires either $\Omega(n^2)$ bits of classical memory or an exponential number (in~$n$) of random samples. A line of recent works continued that research direction and showed that for a large collection of classical learning tasks, either super-linear classical memory size or super-polynomially many samples are needed. All these works consider learning algorithms as classical branching programs, which perform classical computation within bounded memory.

However, these results do not capture all physical computational models, remarkably, quantum computers and the use of quantum memory. It leaves the possibility that a small piece of quantum memory could significantly reduce the need for classical memory or samples and thus completely change the nature of the classical learning task. Despite the recent research on the necessity of quantum memory for intrinsic quantum learning problems like shadow tomography and purity testing, the role of quantum memory in classical learning tasks remains obscure. 

In this work, we study classical learning tasks in the presence of quantum memory. We prove that any quantum algorithm with both, classical memory and quantum memory, for parity learning on $n$ bits, requires either $\Omega(n^2)$ bits of classical memory or $\Omega(n)$ bits of quantum  memory or an exponential number of samples. In other words, the memory-sample lower bound for parity learning remains qualitatively the same, even if the learning algorithm can use, in addition to the classical memory, a quantum memory of size $c n$ (for some constant $c>0$).

Our result is more general and applies to many other classical learning tasks.
Following previous works, we represent by the matrix $M: A \times X \to \{-1,1\}$ the following learning task. An unknown $x$ is sampled uniformly at random from a concept class $X$, and a learning algorithm tries to uncover $x$ by seeing streaming of random samples $(a_i, b_i = M(a_i, x))$ where for every $i$, $a_i\in A$ is chosen uniformly at random.
Assume that $k,\ell,r$ are integers such that any submatrix of $M$ of at least $2^{-k}\cdot|A|$ rows and at least $2^{-\ell}\cdot|X|$ columns, has a bias of at most $2^{-r}$. We prove that any algorithm with classical and quantum hybrid memory for the learning problem corresponding to $M$ needs either (1) $\Omega(k \cdot \ell)$ bits of classical memory, or (2) $\Omega(r)$ qubits of quantum memory, or (3) $2^{\Omega(r)}$ random samples, to achieve a success probability at least $2^{-O(r)}$. 

Our results refute the possibility that a small amount of quantum memory significantly reduces the size of classical memory needed for efficient learning on these problems. Our results also imply improved security of several existing cryptographical protocols in the bounded-storage model (protocols that are based on parity learning on $n$ bits), proving that security holds even in the presence of a quantum adversary with at most $c n^2$ bits of
classical memory and $c n$ bits of quantum memory (for some constant $c>0$). 

\end{abstract}

\section{Introduction}

Memory plays an important role in learning. Starting from the seminal works by Shamir~\cite{shamir2014fundamental} and  Steinhardt, Valiant and Wager~\cite{steinhardt2016memory}, a sequence of works initiates and deepens the study of lower bounds for learning under memory constraints.
%Perhaps the easiest example is parity learning, in which a learning algorithm tries to learn an unknown $n$-bit sttin
Steinhardt, Valiant, and Wager~\cite{steinhardt2016memory} conjectured that 
%even for (perhaps) the most straightforward learning problem, parity learning, insufficient memory leads to the infeasibility of learning. More concretely, they put forward the following conjecture. 
in order to learn an unknown $n$-bit string from samples of random-subset parity, an algorithm needs either memory-size quadratic in $n$ or exponentially many random samples (also in $n$). This conjecture was later on proved by Raz~\cite{R16},
showing for the first time that for some learning problems, super-linear memory size is required for efficient learning. 
This result was then generalized to a broad class of learning problems~\cite{kol2017time,R17,MoshkovitzM18,BeameGY18,GRT18,sharan2019memory,GargKR20,GargKLR21}. 

Although we have a comprehensive understanding of the (in)feasibility of learning under limitations on particular computation resource (memory), the previous works mentioned above do not capture all physical computational models; most notably, quantum computation and the power of quantum memory. 
Many researchers believe that large-scale quantum computers will eventually become viable. 
Recent experiments demonstrated quantum advantages, for example~\cite{arute2019quantum}, and suggested that there are possibly no fundamental barriers to achieving quantum memory and quantum computers. 
Questions on the role of quantum memory in learning were proposed by Wright in the context of general state tomography~\cite{wright2016learn} and by Aaronson for shadow tomography~\cite{aaronson2018shadow}. A line of works~\cite{haah2016sample,bubeck2020entanglement,huang2020predicting,chen2022exponential,aharonov2022quantum,chen2022toward} pioneer the idea and show either polynomial or exponential separations for learning with/without quantum memory, but all for intrinsic quantum learning tasks like state tomography, shadow tomography and purity testing. 

In light of the above, it is appealing 
to consider classical learning tasks in the presence of quantum memory, as well as hybrid classical-quantum memory. 
A direct implication of all aforementioned classical results only gives trivial results. As $k$ qubits of memory can always be efficiently simulated by $\sim 2^k$ classical bits, we can only conclude (say, for parity learning) that either $\sim 2 \log n$-qubit quantum memory or exponentially many samples are needed.
Prior to our work, it could have been the case that even if only a very small size quantum memory was available, it might have significantly reduced the need for classical memory and led to an efficient learning algorithm.

%In this work, we raise and answer the following intriguing questions:

%\begin{quote}
%%    {\it Does quantum (or hybrid) memory help classical learning tasks?}
%    {\it Does fast learning (of classical tasks) require good quantum memory?
%    
%    Does quantum memory reduce the need for classical memory for fast learning? }
%\end{quote}

In this work, we prove memory-sample lower bounds in the presence of hybrid memory for a wide collection of classical learning problems. As in~\cite{R17,GRT18}, we will represent a learning problem by a matrix $M: A \times X \to \{-1,1\}$ whose columns correspond to concepts in the concept class $X$ and rows correspond to random samples. In the learning task, an unknown concept $x \in X$ is sampled uniformly at random and each random sample is given as $(a_i, b_i) = (a_i, M(a_i, x))$ for a uniformly picked $a_i \in A$. The learner's goal is to uncover $x$. In~\cite{GRT18}, it is proved  that when the underlying matrix $M$ is a $(k, \ell)$-$L_2$ two source extractor\footnote{Roughly speaking, this means that every submatrix $M'$ of $M$ with number of rows at least $2^{-k} |A|$ and number of columns at least $2^{-\ell} |X|$  has a relative bias at most $2^{-r}$.} with error $2^{-r}$, a learning algorithm requires either $\Omega(k \cdot \ell)$ bits of memory or $2^{\Omega(r)}$ samples to achieve a success probability at least $2^{-O(r)}$ for the learning task.

\subsection{Our Results}

In this work, we model a quantum learning algorithm as a program with hybrid memory consisting of $q$ qubits of quantum memory and $m$ bits of classical memory. At each stage, a random sample $(a_i, b_i = M(a_i, x))$ is given to the algorithm. The quantum learning algorithm applies an arbitrary quantum channel to the hybrid memory, controlled by the random sample. Although the channel can be arbitrary, we impose the outcome to be a hybrid classical-quantum state of at most $q$ qubits and $m$ bits. 
We stress that there is no limitation on the complexity of the quantum channel (and this only makes our results stronger as we are proving here lower bounds for such algorithms).
%Here we only ensure the input and output behavior of the quantum channel but nothing about the realizability of the channel within the $q$-qubit-$m$-bit hybrid system. 
%We believe the realizability of quantum channels with bounded quantum memory is difficult to prove or disprove, as it would prove lower bounds on bounded-space quantum computation. Our modeling gives stronger memory-sample lower bounds and avoids the difficulty above. 

With the above model, we give the following main theorem. 
\begin{theorem}[Main Theorem, Informal]\label{thm:main_informal}
Let  $M: A \times X \to \{-1,1\}$ be a matrix. If $M$ is a $(k, \ell)$-$L_2$ two source extractor with error $2^{-r}$, a quantum learning algorithm requires either 
\begin{enumerate}
    \item $\Omega(k \cdot \ell)$ bits of classical memory; or,
    \item $\Omega(r)$ qubits of quantum memory; or,
    \item $2^{\Omega(r)}$ samples,
\end{enumerate}
to succeed with a probability of at least $2^{-O(r)}$ in the corresponding learning task. 
\end{theorem}

Our main theorem implies that for many learning problems, the availability of a quantum memory of size up to $\Omega(r)$, does not reduce the size of classical memory or the number of samples that are needed. As coherent quantum memory is challenging for near-term intermediate-scale quantum computers and is probably expensive even if and when quantum computers are widely viable, the impact of quantum memory is further limited for these learning problems. 

To make the theorem more precise, let us take parity learning as an example. The above theorem says that a quantum learning algorithm needs either $\Omega(n^2)$ bits of memory, or $\Omega(n)$ qubits of quantum memory, to conduct efficient learning; otherwise, it requires $2^{\Omega(n)}$ random samples. At first glance, it seems that the constraint on quantum memory is trivial: if the target is to learn an $n$-bit unknown secret, a linear amount of memory always seems necessary to store the secret. However, noticing that our main theorem applies to quantum learning algorithms with hybrid memory and rules out algorithms with $n^2/1000$ bits and $n/1000$ qubits of hybrid memory for parity learning, the main theorem yields non-trivial and compelling memory-sample lower bounds. 
Note also that our results (and previous results) are valid even if the goal is to output only one bit of the secret.
Currently, we do not know whether our main theorem is tight. For parity learning, we are not aware of any quantum learning algorithm that uses only $O(n)$ qubits of quantum memory. We leave closing the gap as a fascinating open question. 

The main theorem naturally applies to other learning problems considered in~\cite{GRT18}, including learning sparse parities, learning from sparse linear equations, and many others. We do not present an exhaustive list here but refer the readers to~\cite{GRT18} for more details.

Along the way, we propose a new approach for proving the classical memory-sample lower bounds. We call this approach, the ``badness levels'' method. The approach is technically equivalent to the previous approach in~\cite{R17,GRT18} but is conceptually simpler to work with and we are able to lift it to the quantum case.
%\qipeng{should we say we simplified the previous proof as a byproduct, somewhere?}

We note that proving a linear lower bound on the size of the quantum memory, without classical memory, is significantly simpler (but to the best of our knowledge such a proof has not appeared prior to our work).
We present such a proof in Appendix~\ref{appendix:C}.
In Appendix~\ref{appendix:C}, we state and prove Theorem~\ref{thm:C} that shows 
a simpler proof for a linear lower bound on the quantum-memory size (without classical memory).
While Theorem~\ref{thm:C} is qualitatively weaker than our main result in most cases, as it only gives a lower bound for programs with only quantum memory but without a (possibly quadratic) classical memory, Theorem~\ref{thm:C} is technically incomparable and is stated in terms of quantum extractors, rather than classical extractors.
Additionally, the proof of Theorem~\ref{thm:C} is significantly simpler than the proof of our main theorem.

\paragraph{Implications to Cryptography in the Bounded-Storage Model.}
Since learning theory and cryptography can be viewed as two sides of the same coin, our theorem also lifts the security of many existing cryptographical protocols in the bounded-storage model (protocols that are based on parity learning) to the quantum setting. To our best knowledge, these are the first proofs of classical cryptographical protocols being secure against space-bounded quantum attackers.\footnote{On the other hand, there are known examples of classically-secure bounded-storage protocols that are breakable with an exponentially smaller amount of quantum memory.
 \cite{Gavinsky08}.} We elaborate more below.

Cryptography in the (classical) bounded storage model was first proposed by Maurer~\cite{maurer1992conditionally}. In such a model, no computational assumption is needed. Honest execution is performed through a long streaming sequence of bits. Eavesdroppers have bounded storage and limited capability of storing conversations, thus cannot break the protocol. A line of works~\cite[\dots]{cachin1997unconditional,aumann1999information,aumann2002everlasting,ding2002hyper,dziembowski2002tight,lu2002hyper,dziembowski2004generating,moran2004non,harnik2006everlasting,dodis2021speak,dodis2022authentication} builds  efficient and secure protocols for key agreement, oblivious transfer, bit commitment and time stamping in that model. 

Based on the memory-sample lower bounds for parity learning of $n$ bits, \cite{R16} suggested an encryption scheme in the bounded-storage model.
Guan and Zhandry~\cite{guan2019simple} proposed key agreement, oblivious transfer and bit commitment with improved rounds and better correctness, against attackers with up to $O(n^2)$ bits of memory. Following a similar idea, Liu and Vusirikala~\cite{liu2021secure} showed that semi-honest multiparty computation could be achieved against attackers with up to $O(n^2)$ bits of memory. More recently, Dodis, Quach, and Wichs 
\cite{dodis2022authentication} considered message authentication in the bounded storage model based on parity learning. Our result on parity learning gives a direct lift on all the results above. When the cryptographic protocols are based on parity learning of $n$ bits (often treated as a security parameter), our result shows that security holds even in the presence of a quantum adversary with at most $O(n^2)$ bits of classical memory and $O(n)$ qubits of quantum memory. %for the construction in~\cite{guan2019simple,liu2021secure}, and similar security holds for the construction~\cite{dodis2022authentication}.

Despite many previous works on cryptography in the quantum bounded storage model~\cite{damgaard2007tight,damgaard2007secure,schaffner2007cryptography,damgaard2008cryptography,wehner2008composable,pironio2013security,broadbent2021device}, they all rely on streaming quantum states. Our memory-sample lower bounds give for the first time a rich class of classical cryptographical schemes (key agreement, oblivious transfer, and bit commitment) secure against space-bounded quantum attackers.

%\cite{konig2008bounded}

	\section{Proof Overview}
	
	\subsection{Recap of Proofs for Classical Lower Bounds}
	Since our proof builds on the previous line of works on classical memory-sample lower bounds for learning, specifically, on the proof technique of \cite{R17,GRT18}, we provide a brief review of these proofs, using parity learning \cite{R16} as an example. In below, $M(a,x)$ denotes the inner product of $a$ and $x$ in $\mathbb{F}_2$.
	
	Consider a classical branching program that tries to learn an unknown and uniformly random $x\in\{0,1\}^n$ from samples $(a,b)$, where $a\in\{0,1\}^n$ is uniformly random and $b=M(a,x)$. We can associate every state $v$ of the branching program with a distribution $P_{X|v}$ over $\{0,1\}^n$, indicating the distribution of $x$ conditioned on reaching that state. At the initial state, without any information about $x$, the distribution is uniform (which has the smallest possible $\ell_2$-norm). Along a computational path on the branching program, the distribution $P_{X|v}$ evolves and eventually gets concentrated (with large $\ell_2$-norms) in order to output $x$ correctly. Therefore, during the evolution, $P_{X|v}$ should at some stage have mildly large $\ell_2$-norms ($2^{\varepsilon n}$ times larger than uniform for some small constant $\varepsilon>0$). If we set such a distribution as a target, the distribution is hard to achieve with random samples. Only with $2^{-\Omega(n)}$ probability, the branching program can make significant progress towards the target; while most of the time a sample just splits the distributions (both the current and the target distribution) into two even parts, and that does not help much in getting closer to the target distribution (with large $\ell_2$ norm).
	
	To put it more rigorously, we examine the evolution of the inner product
	\[\langle P_{X|v}, P\rangle=\sum_{x\in\{0,1\}^n} P_{X|v}(x)\cdot P(x)\]
	between the distribution $P_{X|v}$ on the current state $v$, and a target distribution $P$. Receiving a sample $(a,b)$ implies that $M(a,x)=b$, hence only the part of $P_{X|v}$ supported on such $x$ proceeds. If this part is close to $\frac{1}{2}$ probability, we say that $a$ divides $P_{X|v}$ evenly. Denoting the new distribution as $P_{X|v}^{(a,b)}$, after proper normalization the new inner product is
	\begin{equation}\label{eq:intuition}
	    \langle P_{X|v}^{(a,b)}, P\rangle
	    =\sum_{\substack{x\in\{0,1\}^n\\M(a,x)=b}} P_{X|v}(x)\cdot P(x)\bigg/
	    \sum_{\substack{x\in\{0,1\}^n\\M(a,x)=b}} P_{X|v}(x).
	\end{equation}
	Ideally, both $P_{X|v}$ and the point-wise product vector $P_{X|v}\cdot P$ should have reasonably small $\ell_2$-norms. Due to the extractor property of $M$, most of $a\in\{0,1\}^n$ should divide both vectors evenly, and thus the denominator is close to $\frac{1}{2}$ while the enumerator is close to $\frac{1}{2}\langle P_{X|v}, P\rangle$. That means, given a uniformly random $a$, we get limited progress on the inner product. On the other hand, from $\langle U,P\rangle=2^{-n}$ with uniform distribution $U$ to $\langle P,P\rangle=2^{2\varepsilon n}\cdot 2^{-n}$, the branching program needs to make multiple steps of progression. Therefore it happens with an extremely small probability.
	
	To ensure that the above statement goes smoothly, we require the following properties for every state $v$ in the branching program:
	\begin{itemize}
	    \item The $\ell_2$-norm $\lN{P_{X|v}}_2$ is small.
	    \item The $\ell_2$-norm $\lN{P_{X|v}\cdot P}_2$ is small, which is implied when the $\ell_\infty$-norm $\lN{P_{X|v}}_\infty$ is small.
	    \item The denominator in \cref{eq:intuition} is bounded away from $0$ for every sample $(a,b)$.
	\end{itemize}
	These properties do not hold by themselves. Instead, we execute a \emph{truncation} procedure on the branching program \emph{before} choosing a target distribution. More specifically, the branching program is modified so that it stops whenever it:
	\begin{itemize}
	    \item ($\ell_2$ truncation): Reaches a state $v$ with large $\lN{P_{X|v}}_2$;
	    \item ($\ell_\infty$ truncation): Reaches a state $v$ with large $P_{X|v}(x)$ when the unknown concept is $x$;
	    \item (Sample truncation): Or, for the next sample $(a,b)$, $a$ does not divide $P_{X|v}$ evenly.
	\end{itemize}
	It turns out that after $\ell_2$ truncation, the other two truncation steps add $2^{-\Omega(n)}$ error in each stage of the branching program. Therefore the proof boils down to proving a $2^{-\Omega(n^2)}$ bound on the probability of reaching a state with large $\lN{P_{X|v}}_2$,
	from which by a standard union bound, we can prove the memory-sample lower bounds for parity learning: either $2^{\Omega(n)}$ samples or $\Omega(n^2)$ bits of memory are necessary. 
	
	\subsection{Badness Levels}
	As mentioned above, to bound the probability of reaching a state with a large $\ell_2$-norm, the basic idea is to fix its distribution as the target distribution $P$, and bound the increment of the inner product $\langle P_{X|v}, P\rangle$. This was done in \cite{R17,GRT18} by designing a potential function that tracks the average of $\langle P_{X|v}, P\rangle^k$ for some $k=\Theta(n)$, where the average is over states $v$ in the same stage of the branching program. Here we propose another approach using the concept of \emph{badness levels}. Although it is technically equivalent to the potential function approach in the classical case, it is more pliable and easier to be adapted to the quantum case. We view this approach as a separate contribution of our work.
	
	We first define a \emph{bad event} to be a pair $(v,a)$ of the state $v$ and the upcoming part of the sample~$a$, such that $\langle P_{X|v}, P\rangle\geq 2^{-n}$, and for one of the two possible outcomes $b$,
	\begin{equation} \label{eq:badevent}
	    \sum_{\substack{x\in\{0,1\}^n\\M(a,x)=b}} P_{X|v}(x)\cdot P(x)\geq 
	\left(\frac{1}{2}+2^{-\delta n}\right)\cdot \langle P_{X|v}, P\rangle
	\end{equation}
	with some small constant $\delta$. In other words, the inner product $\langle P_{X|v}, P\rangle$ is large enough, while not being divided evenly by $a$. From \cref{eq:intuition} we know that the inner product gets at most roughly doubled through a bad event. In contrast, in the good case, the inner product either gets a mere $(1+2^{-\delta n})$ multiplicative factor or is already smaller than the baseline $2^{-n}$. Also, the extractor property of $M$ ensures that for every state $v$, over uniformly random $a$, the bad event happens with at most $2^{-\Omega(n)}$ probability. 
	
	Now, the badness level $\beta(v)$ of a state $v$ keeps track of how many times the computational path went through bad events before reaching $v$.\footnote{For now we think of $\beta(v)$ as a natural number. In the actual proof, $\beta(v)$ is a distribution on natural numbers, as for different computational paths reaching the same state, the count of bad events can be different.} The above observations on the bad events imply that (omitting the smaller factors):
	\begin{itemize}
	    \item For every state $v$, $\langle P_{X|v}, P\rangle$ is bounded by $2^{\beta(v)}\cdot 2^{-n}$;
	    \item Heading to the next stage, $\beta(v)$ increases by $1$ with probability $2^{-\Omega(n)}$.
	\end{itemize}
	Therefore at each stage, the total weight of states with badness level $\beta$ is at most $2^{-\Omega(\beta n)}$. Thus any state with $\langle P_{X|v}, P\rangle\geq 2^{2\varepsilon n}\cdot 2^{-n}$ must have $2^{-\Omega(n^2)}$ probability.
	
	\subsection{Obstacles for Proving Quantum Lower Bounds}
	%Turning to the quantum case,
    In this section, we present an attempt to prove the same $2^{\Omega(n)}$-sample or $\Omega(n^2)$-quantum-memory lower bound for the pure quantum case. Along the way we identify some obstacles to proving memory-sample lower bounds for \emph{quantum} learning algorithms, and in the next section we show how to overcome these obstacles  while proving lower bounds for \emph{hybrid} learning algorithms, with quadratic-size classical-memory and linear-size  quantum-memory.
    
    Following the same framework as the above described proof for the classical case, we first need to transfer all the notions to a quantum algorithms:

	%Ambitiously we hope to prove the same $2^{\Omega(n)}$-sample or $\Omega(n^2)$-quantum-memory lower bound,  
	%by following the same framework as the above described proof for the classical case. We need to first transfer all the above notions to a quantum branching program:
	\begin{itemize}
	    \item The state $v$ is a quantum state in the Hilbert space of quantum memory;
	    \item The distribution $P_{X|v}$ is still well-defined: It is the distribution of $x$ when the quantum memory is measured to $v$ (see \cref{sec:cq-system} and \cref{eq:induced_prob});
	    \item We are  still able to implement $\ell_2$ truncation: If $P_{X|v}$ has large $\ell_2$-norm, project the entire system to the orthogonal subspace $v^\perp$ of $v$ and repeat, until there is no such state $v$ (see \cref{sec:trunc} for details).
	    \item We are also able to implement sample truncation, in a similar manner to $\ell_2$ truncation. As the criteria here depends on $a$, we separately create a copy of the current system for each~$a$, truncate the states $v$ using projection when $P_{X|v}$ is not evenly divided by $a$ in each copy, and then merge them back together. We prove that the error introduced by this truncation is small.
	\end{itemize}
	
	%Here comes the first major obstacle: $\ell_\infty$ truncation. In the classical case, the reason that $\ell_\infty$ truncation needs to be implemented for each individual $x$ (in contrast to $\ell_2$ truncation where we simply remove the states altogether) is that with large enough memory, all states might simultaneously have large $\ell_\infty$-norms, as shown in the following example:
	
	%\begin{example}
	%    Consider a classical learning algorithm whose memory stores a hypothetical $x$ by random guessing. It then checks consistency with each upcoming sample: if consistent, the memory is kept unchanged; otherwise the algorithm makes another random guess and stores the new guess in the memory. Along the computation, the distribution on each memory state gets more and more concentrated on the corresponding $x$. At some stage they simultaneously pass the $\ell_\infty$-norm threshold, and everything gets removed if we execute $\ell_\infty$ truncation the same way as $\ell_2$ truncation.
	%\end{example}
	
	Here comes the first major obstacle: $\ell_\infty$ truncation. In the classical case,  $\ell_\infty$ truncation is implemented for each individual $x$, in contrast to $\ell_2$ truncation where  the states are removed altogether.  Relying on the fact that it is already known that the $\ell_2$ norm of the distribution is small, using Markov inequality, one can prove that the error introduced by the $\ell_\infty$ truncation is small. 
	
	However, when we try to emulate the classical implementation of $\ell_\infty$ truncation with quantum truncation, that is, to only project to $v^\perp$ the system \emph{conditioned on} the specific $x$ where $P_{X|v}(x)$ is large, instead of for every $x$, it may lead to huge changes to the distributions $P_{X|u}$ on states $u$ non-orthogonal to $v$. The following example illustrates such a scenario:
	
	\begin{example}
	    Consider a quantum learning algorithm, and assume that at some stage of the computation, for each $x\in\{0,1\}^n$, the quantum memory is in some pure state $v(x)$. We pick each $v(x)$ uniformly at random in a Hilbert space of dimension $d\approx 2^{n/2}$ and consider a typical configuration of $v(x)$. Now the $\ell_2$-norms are bounded for every quantum state $v$: the worst ones happen when $v=v(x)$ for some $x$, where $\|P_{X|v(x)}\|_2$ is typically around $d\cdot 2^{-n}$, close to the $\ell_2$-norm of uniform distribution. However, those worst distributions also have $\ell_\infty$-norms close to $d\cdot 2^{-n}$, which is much larger than the $\ell_\infty$-norm of the uniform distribution, and needs to be truncated. But truncating $v(x)$ off for $x$ means that $x$ is completely erased, and we end up removing everything.
	\end{example}
	
	Moving on, we fix a target state $v$ with a target distribution $P_{X|v}$ which exceeds the $\ell_2$-norm threshold, and the goal is again to prove a $2^{-\Omega(n^2)}$ amplitude bound on $v$. The bad event should still be defined as a pair $(v,a)$ satisfying \cref{eq:badevent}, with $v$ now being a quantum state. We then run into the second major obstacle: it is not clear how to define badness levels.
	
	If we define the badness level $\beta(v)$ for each state $v$ individually by examining the bad events over the historical states, then it is not clear how to measure the total weight of a badness level~$\beta$. In the classical case, we simply define the total weight as the total probability of states with badness level $\beta$. But here in the quantum case, it turns out that such a definition either depends on the choice of basis, which might have large increment in each stage, or completely fails to imply the desired amplitude bound on the target state.
	
	% Here we discuss two such choices and the problems therein. For simplicity of discussion, let us assume that on sample $(a,b)$ the quantum branching program applies a unitary operator $U_{a,b}$, in contrast to any quantum channel.

    % The first obvious choice is to define the badness level $\beta(v)$ at stage $t$ for each state $v$ individually, by examining the bad events over the historical states $U_{a_{t-1},b_{t-1}}^{-1}v$ and so on. The problem with this definition is that $\beta(v)$ gets too arbitrary across the Hilbert space: Even when two states $v_1$ and $v_2$ are close, the relationship between $\beta(v_1)$ and $\beta(v_2)$ is not definite, plus the total weight of a badness level $\beta$ depends on the choice of basis. Specifically, to understand the badness levels on a basis $\mc B$, we need to examine the basis $U_{a,b}^{-1}\mc B$ in the previous stage if the last sample is $(a,b)$, and there could always be some $v\in U_{a,b}^{-1}\mc B$ such that $\beta(v)$ was high and $(v,a)$ is bad. Indeed,  $(v,a')$ is good for most other $a'$, and that is the key for diluting the increment of badness level in the classical case. But here, as far as we know, $U\mt_{a',b'}v$ could be totally irrelevant to $\mc B$.

    The other choice is to have a more \emph{operational} definition of badness levels, and it is indeed tempting to define $\beta$ as another register whose updates are controlled by the quantum memory. The problem with such definitions is that the bad event (\cref{eq:badevent}) is not linear in $v$. Therefore an operational definition of badness level, which is a linear operator, inevitably introduces error that escalates fast with the number of stages.
    
    \subsection{Hybrid Memory Lower Bounds with Small Quantum Memory}
    The obstacles in the previous section are for proving quadratic quantum memory lower bound. We note that proving linear quantum memory lower bound (without classical memory) is not hard: the proof can be entirely information theoretical, as with very limited memory, say, $\frac{1}{2}n$ qubits, the information gained from each sample is exponentially small, despite the memory being quantum. We present such a proof in \cref{appendix:C}.
    
    The lower bounds that we prove here are with hybrid memory: To learn parity with both classical and quantum memory, an algorithm needs either $2^{\Omega(n)}$ samples, or $\Omega(n^2)$ classical memory, or $\Omega(n)$ quantum memory (\cref{thm:main_informal}). We now describe how we overcome the previously mentioned obstacles.
    
    \paragraph{$\ell_\infty$ Truncation.} When there is only small quantum memory and no classical memory, the treatment for $\ell_\infty$ truncation is straightforward. We remove all quantum states $v$ with distributions of large $\ell_\infty$-norm, by projecting the system to the orthogonal subspace $v^\perp$, just like the process of $\ell_2$ truncation. As the overall distribution on $x$ is uniform, any state $v$ with $\|P_{X|v}\|_\infty\geq 2^{\delta n}\cdot 2^{-n}$ must have weight at most $2^{-\delta n}$. Therefore, as long as the dimension of the Hilbert space is much smaller than $\delta n$, the error introduced in this truncation is small. \footnote{The example in the previous section that shows the infeasibility of treating $\ell_\infty$ truncation the same way as $\ell_2$ truncation does not work here, as it requires $n/2$ qubits of memory while here we have a smaller memory size.}
    
    With classical memory in presence, the actual $\ell_\infty$ truncation step (see \cref{sec:truncd}, Step 2) is more complicated. We first apply the original classical $\ell_\infty$ truncation on the classical memory $W$. Now that $\|P_{X|w}\|_\infty$ is bounded for each classical memory state $w$, we can remove the quantum states $v$ with large $\|P_{X|v,w}\|_\infty$ by projection as stated above. Since the classical $\ell_\infty$ truncation depends on $x$, it could change the distributions $P_{X|v,w}$. However, as in the classical case, $P_{X|w}$ will not change a lot. Thus, wherever $P_{X|v,w}$ changes drastically, it must have a small weight and can also be removed by projection. This removal corresponds to truncation by $G_t$ in \cref{sec:truncd}.
    
    \paragraph{Badness Levels.} Interestingly, we are able to avoid the problems of defining the badness level on quantum memory altogether, by keeping it a property on the classical memory only. To do so we need to alter the definition of a bad event: it is now a pair $(w,a)$ of classical memory state $w$ and sample $a$, such that \emph{there exists} some quantum memory state $v$ with $P_{X|v,w}$ satisfying \cref{eq:badevent}.
    
    For each fixed classical memory state $w$, we still need to ensure that bad events happen with a small probability. We prove it (\cref{lemma:badness}) by showing that, if there are many different samples~$a$, each associated with some quantum state $v_a$ satisfying \cref{eq:badevent}, then there is some quantum state~$v$ that simultaneously satisfies \cref{eq:badevent} with most of such $a$ (which is impossible because of the extractor property). This is ultimately due to the continuous nature of \cref{eq:badevent}: Under some proper congruent transformation, \cref{eq:badevent} becomes a simple threshold inequality on quadratic forms over~$v$. Now if it is satisfied by some $v_a$, it is going to be satisfied by most~$v$ for a much smaller threshold parameter $\delta$, and hence the existence of a simultaneously satisfying $v$.\footnote{We note that the error bound for sample truncation (\cref{lemma:g_a}) is also proved using this argument.} In this argument, we use \cref{lemma:anticc}, which is derived from the anti-concentration bound for Gaussian quadratic forms, and crucially relies on the fact that the dimension is at most $2^{\varepsilon n}$ for some small $\varepsilon$.
    
    Another technical problem is that to use the extractor property, we need to ensure that $\langle P_{X|v,w}, P\rangle\geq 2^{-n}$ for the simultaneously satisfying $v$. Thus, what we do in \cref{lemma:badness} is to first conceptually remove the parts where $\langle P_{X|v,w}, P\rangle$ is too small, using projection similarly to the truncation steps. After the removal, we are left with a subspace $\mc V'$ where $\langle P_{X|v,w}, P\rangle$ is always lower bounded, and we show that for every state $v$ that satisfies \cref{eq:badevent}, the inequality is still close to being satisfied after projecting $v$ onto $\mc V'$. Therefore we could still apply the above argument and find a simultaneously satisfying $v$ within the subspace.
	
	\section{Preliminaries}
	
	\subsection{Vectors and Matrices}
	
	%\qipeng{red color for comments, }
	%\revise{brown color for revised contents}
	
	For a vector $v\in\mathbb{C}^d$ and $p\in[1,\infty]$, we define the $\ell_p$ norm of $v$ as
	\[\|v\|_p=\left(\sum_{i=1}^d |v_i|^p\right)^{1/p}.\]
	For two vectors $u,v\in\mathbb{C}^d$, define their inner product as $\langle u,v\rangle = u\ct v=\sum_{i=1}^d \overline{u_i}v_i$. So $\|v\|_2^2=\langle v,v\rangle$. We also view every distribution $P$ over a set $\mc X$ as a non-negative real vector with $\|P\|\mt_1=1$.
	
	We specifically use Dirac notation to denote unit vectors, $\ket{v}\in\mathbb{C}^d$ implies that $\|\ket{v}\|_2=1$. For a non-zero vector $u\in\mathbb{C}^d$, let $\ket{v}\sim u$ be the normalization of $u$, that is, $\ket{v}=u/\|u\|_2$.
	
	For every vector $v\in\mathbb{C}^d$, let $\Diag v\in\mathbb{C}^{d\times d}$ be the diagonal matrix whose diagonal entries represent $v$. Conversely, for every square matrix $M$, let $\diag M$ be the vector consisting of the diagonal entries of $M$. For a matrix (or generally a linear operator) $M$,  we use $\|M\|\mt_\Tr$ and $\|M\|_2$ to denote its trace norm and spectral norm respectively, that is,
	\[\lN{M}_\Tr=\Tr\left[\sqrt{M M\ct}\right],\quad \lN{M}_2=\max_{v \ne 0} \|M v\|_2/\|v\|_2.\]
	For an Hermitian $M \in \mathbb{C}^{d \times d}$, we say it is a positive semi-definite operator if for every $v \in \mathbb{C}^d$, $v\ct M v \geq 0$. A (partial) density operator is a positive semi-definite operator with its trace being $1$ (or at most $1$, respectively). 
	%\qipeng{some notations that people should be familiar with, but not mentioned here (used in the main body):  chi-squared distribution, trace out, adding more. }
	
	\subsubsection*{Viewing a Learning Problem as a Matrix}
	
	Let $M:\mc A\times\mc X\rightarrow\{-1,1\}$ be a matrix. The matrix $M$ corresponds to the following
	learning problem. There is an unknown element $x\in \mc X$ that was chosen uniformly at random.
	A learner tries to learn $x$ from samples $(a, b)$, where $a\in\mc A$ is chosen uniformly at random and
	$b = M(a, x)$. That is, the learning algorithm is given a stream of samples, $(a_1,b_1),(a_2,b_2),\ldots$, where each $a_t$ is uniformly distributed and for every $t$, $b_t = M(a_t, x)$. For each $a\in\mc A$, we use $M_a:\mc X\rightarrow\{-1,1\}$ to denote the vector corresponding to the $a$-th row of $M$.
	
	\subsubsection*{Extractors}
	
	A matrix $M:\mc A\times\mc X\rightarrow\{-1,1\}$ with $n=\log_2|\mc X|$ is a $(k,\ell)$-$L_2$ extractor with error $2^{-r}$, if for every distribution $P$ over $\mc X$ with $\|P\|_2\leq 2^\ell\cdot 2^{-n/2}$, there are at most $2^{-k}\cdot|\mc A|$ rows $a\in\mc A$ such that
	\[|\langle M_a, P\rangle|\geq 2^{-r}.\]

	%\subsection{Other Useful Theorems}
	
	%Choi's Theorem, not needed for now
	%\revise{
	%\subsubsection*{Choi's Theorem}
	%\begin{theorem}[\cite{Choi1975}]\label{thm:choi}
	%    Let $\Phi$ be a linear map from $\mathbb{C}^{n \times n} \to \mathbb{C}^{m \times m}$. The following two conditions are equivalent: 
	%    \begin{enumerate}
	%        \item $\Phi$ is positive;
	%        \item The following operator (called the Choi matrix of $\Phi$) is positive: 
	%        \begin{align*}
	%            \mathfrak{J}(\Phi) := \sum_{i, j} \Phi(E_{i, j}) \otimes E_{i, j}, 
	%        \end{align*}
	%        where $E_{i, j}$ is the matrix with $1$ in the $(i,j)$-th entry and $0$s elsewhere.
	%    \end{enumerate}
	%\end{theorem}
	%}
	
	\subsection{Anti-Concentration Bound for Quadratic Form on Unit Vectors}
	\begin{lemma}\label{lemma:anticc}
		There exists an absolute constant $c$ such that following holds. Let $\sigma$ be a Hermitian operator over the Hilbert space $\mc V=\mathbb{C}^d$, and let $v$ be a uniformly random unit vector in $\mc V$. Then for every $\varepsilon>0$, we have
		\[\Pr\left[|v\ct\sigma v|\leq \frac{\varepsilon\|\sigma\|_2}{d}\right]\leq c\sqrt{\varepsilon}+e^{-d}.\]
	\end{lemma}
	\begin{proof}
		Let $g=(g_1,\ldots,g_d)\sim\mathcal{N}(0,1)^d$ be standard Gaussians. Notice that $\|g\|_2^2$ follows $\chi_d^2$-distribution, and $|g\ct\sigma g|/\|g\|_2^2$ is equidistributed as $|v\ct\sigma v|$. Therefore by union bound we have
		\[\Pr_v\left[|v\ct\sigma v|\leq \frac{\varepsilon\|\sigma\|_2}{d}\right]=
		\Pr_g\left[|g\ct\sigma g|\leq {\varepsilon\|\sigma\|_2} \cdot \frac{\|g\|^2_2}{d}\right] \leq 
		\Pr_g\left[|g\ct\sigma g|\leq 5\varepsilon\|\sigma\|_2\right]+
		\Pr_g\Big[\|g\|_2^2\geq 5d\Big].\]
		For the first term, notice that $\Var[g\ct\sigma g]=2\Tr[\sigma^2]$ (see e.g. \cite[Chapter~5]{rencher2008linear}) which is no smaller than $2\|\sigma\|_2^2$. Therefore, by Carbery–Wright inequality \cite{Carbery2001}, there exists an absolute constant $c$ such that
		\[\Pr\left[|g\ct\sigma g|\leq 5\varepsilon\|\sigma\|_2\right]\leq 
		\Pr\left[|g\ct\sigma g|\leq 4\varepsilon\Var[g\ct\sigma g]^{1/2}\right]\leq c\sqrt{\varepsilon}.\]
		For the second term, the standard Laurent-Massart bound on $\chi^2$-distrbitions \cite{Laurent2000} gives:
		\[\Pr\Big[\|g\|_2^2\geq 5d\Big]\leq e^{-d}. \qedhere\]
		%\qipeng{should we explicitly list the statement in CW01 and LM00?}
	\end{proof}
	
	\subsection{Multipartite Quantum Systems}
	The state of $q$ qubits can be represented in a Hilbert space $\mc V=(\mathbb{C}^2)^{\otimes q}=\mathbb{C}^{2^q}$. In a product of $m$ Hilbert spaces $\mc V_{[m]}=\mc V_1\otimes \cdots\otimes \mc V_m$, a multipartite partial system $V_1,\ldots,V_m$ is represented by a partial density operator $\rho\mt_{V_{[m]}}$. For a subset $I\subseteq[m]$ of indices, the subsystem on $\{V_i\}_{i\in I}$ (or $V_{I}$ for short) is defined by tracing out $j\notin I$, that is,
	\[\rho\mt_{V_{I}}=\Tr\mt_{V_{j\notin I}}[\rho\mt_{V_{[m]}}].\]
	Now for any two disjoint subsets $I,J\subset[m]$, given some $\ket{v_J}\in\mc V_J=\bigotimes_{j\in J}\mc V_j$, the conditional system on $V_I$ is defined as
	\[\rho\mt_{V_{I}|v_{J}}=
	\left(\id_{V_{I}} \otimes \bra{v_J}\right)
	\rho\mt_{V_{I\cup J}}
	\left(\id_{V_{I}}\otimes \ket{v_J}\right),\]
	which is a partial density operator on $V_{I}$. Note that the trace \[\Tr\left[\rho\mt_{V_{I}|v_{J}}\right]= \bra{v\mt_J} 
	\rho\mt_{V_{J}}
	 \ket{v\mt_J}\]
	only depends on the system $\rho$ and $\ket{v_J}$, while being \emph{independent} of the choice of $I$. 
	%Therefore, we will sometimes {\bf ignore $I$} and only write the trace as $\Tr[\rho_{|v_{J}}]$.
	
	Another simple fact that will be repeatedly used later on is that for an \emph{orthogonal basis} $\mc B$ of $V_J$, we have
	\[\rho\mt_{V_{I}}=\Tr_{V_{J}}[\rho\mt_{V_{I\cup J}}]
	=\sum_{\ket{v_{J}}\in\mc B}\rho\mt_{V_{I}|v_{J}}.\]
	
	\subsection{Classical-Quantum Systems}\label{sec:cq-system}
	In the underlying space $\mc V_1\otimes \cdots\otimes \mc V_m$ of the multipartite system, we say $\mc V_i$ is classical if there is a fixed orthogonal basis $\mc B_i$ of $\mc V_i$, such that for every multipartite system $\rho\mt_{V_{[m]}}$, every pair of distinct $\ket{v_i}\neq \ket{v_i'}\in\mc B_i$ and every two states $\ket{v},\ket{v'}\in \bigotimes_{j\neq i}\mc V_j$, we have
	\[\bra{v_i,v}\rho\mt_{V_{[m]}}\ket{v_i',v'}=0.\]
	Without loss of generality, in the rest of the work we always assume $\mc B_i$ is the set of computational basis states. We also identify $\mc V_i$ with the discrete set $\mc B_i$, and remove the Dirac brackets when we talk about the classical elements in $\mc V_i$. In this case every multipartite system $\rho\mt_{V_{[m]}}$ can be written as a direct sum
	\[\rho\mt_{V_{[m]}}=\bigoplus_{v_i\in\mc V_i}\rho\mt_{V_{[m]\setminus\{i\}}|v_i}.\]
	The reader may find this direct sum viewpoint easier to handle in some later scenarios.
	
	When $V_{I}$ is classical, conditioned on any $\ket{v_J}\in\mc V_J$ with $J$ disjoint from $I$, the system $\rho\mt_{V_{I}|v_{J}}$ is represented as a diagonal matrix on $V_I$. If $\Tr[\rho\mt_{V_{I}|v_{J}}]>0$, it induces a distribution over \
	the computation basis states of $V_I$, defined as
    \begin{align}\label{eq:induced_prob}
	P^\rho_{V_{I}|v_{J}}=\diag\rho\mt_{V_{I}|v_{J}}/\Tr[\rho\mt_{V_{I}|v_{J}}]. % = \diag\rho\mt_{V_{I}|v_{J}}/\Tr[\rho\mt_{|v_{J}}].
	\end{align}
	In the rest of this paper, whenever we use this notation $P^\rho_{V_{I}|v_{J}}$, it is always implicitly assumed that $\Tr[\rho\mt_{V_{I}|v_{J}}]>0$ and the distribution exists.
	
	In this work we typically consider the following scenario: There is a quantum memory register $V$ ranging in the Hilbert space $\mc V$, and a classical memory register $W$ ranging in the set of memory states $\mc W$, along with some classical information $X\in\mc X$ (later in the work, it is the concept to be learned) that is correlated with $V$ and $W$. We will make use of the following fact:
	
	\begin{claim}\label{claim:lotp}
	    Let $\rho_{XVW}$ be a classical-quantum system over classical $X, W$ and quantum $V$. 
		For every $w\in\mc W$, $P^\rho_{X|w}$ is a convex combination of $P^\rho_{X|v,w}$ for some $\{\ket{v}\}\subseteq\mc V$.
	\end{claim}
	\begin{proof}
	    Let $\mc B$ be an orthogonal basis of $\mc V$, so that we have (from the end of last section)
	    \[
	        \rho\mt_{X|w} =\sum_{\ket{v}\in \mc B}\rho\mt_{X|v,w}.
	    \]
	    Therefore $P^\rho_{X|w}$ is a linear combination of $P^\rho_{X|v,w}$ for $\ket v \in \mc B$, with non-negative coefficients. Since they are all distributions, it must be a convex combination.
	\end{proof}
	
	\paragraph{Characterization of operators over classical-quantum hybrid systems.}
	
	Now we identify all possible operators on the classical-quantum hybrid memory space $\mc V\otimes \mc W$. A priori to the assumption that $W$ is classical, we think of a quantum channel operating on the system as working on the underlying space $\mc V\otimes\mathbb{C}^{|\mc W|}$. Now we denote $\mathscr{T}\mt_{\mc V\otimes\mc W}$ to be the set of all such quantum channels $\Phi$ that satisfy the following: for every classical-quantum system $\rho\mt_{VW}$ in $\mc V\otimes\mc W$, $W$ is still classical in $\Phi(\rho\mt_{VW})$. That is, for every two states $\ket{v},\ket{v'}\in \mc V$ and every pair of distinct $w,w'\in\mc W$, we have
	\[\bra{v,w}\Phi(\rho\mt_{VW})\ket{v',w'}=0.\]
	
	Note that not all channels in $\mathscr{T}\mt_{\mc V\otimes\mc W}$ are physically realizable. For instance, with one-bit classical memory and no quantum memory, the channel
	\[\begin{pmatrix}a & c \\ \overline{c} & b\end{pmatrix}\mapsto
	\begin{pmatrix}a & ic \\ -i\overline{c} & b\end{pmatrix}\]
	is not a classical operator. However, since we are constrained to classical quantum systems, this channel is effectively equivalent to an identity channel on one-bit classical memory. Generally speaking, every channel in $\mathscr{T}\mt_{\mc V\otimes\mc W}$ is equivalent to a channel controlled by $\mc W$ that maps $\mc V$ to $\mc V\otimes\mc W$. Below, we prove this observation and use it to show the following claim:
	
	\begin{claim}\label{claim:chan} 
	    Let $\rho_{XVW}$ be a classical-quantum system over classical $X, W$ and quantum $V$.
		Let $\Phi\in\mathscr{T}\mt_{\mc V\otimes\mc W}$, and we use $\Phi(\rho)$ to denote the system after applying $\Phi$ to $VW$ and identity to $X$. Then for every $\ket{v}\in\mc V$ and $w\in\mc W$, $P^{\Phi(\rho)}_{X|v,w}$ is a convex combination of $P^{\rho}_{X|v',w'}$ for some $\{\ket{v'}\}\subseteq\mc V$ and $\{w'\}\subseteq\mc W$. 
	\end{claim}
	One difference between \Cref{claim:lotp} and \Cref{claim:chan} is that in \Cref{claim:chan} it is not always possible to write $P^{\Phi(\rho)}_{X|v,w}$ as a convex combination of $P^{\rho}_{X|v',w'}$ for $\ket{v'}$ from an orthogonal basis of $\mc V$. But it is always possible in \Cref{claim:lotp}.  Although the difference does not matter in this work, we mention it here for clarity. 
	
	\begin{proof}
	    Since $\Phi \in\mathscr{T}\mt_{\mc V\otimes\mc W}$, the following channel is functionally equivalent to $\Phi$ for classical-quantum systems: 
	    \begin{align*}
	        \Phi' : \rho \to \sum_{w \in \mc W} \Phi(\rho_{V|w} \otimes \ket w \bra w).
	    \end{align*}
	    The physical meaning of $\Phi'$ is to measure $W$ under the computational basis (which should not change the functionality we care about) and apply $\Phi$. 
	    
	    By defining the channel $\Phi_w(\cdot) := \Phi(\cdot \otimes \ket w \bra w)$, 
	    the above can be alternatively written as: 
	    \begin{align*}
	        \Phi' : \rho \to \sum_{w \in \mc W} \Phi_w(\rho_{V|w}).
	    \end{align*}
	    
	    Now consider the Kraus representation of each $\Phi_w$, that is, a finite set of linear operators $E_{w,k}:\mc V\rightarrow\mc V\otimes\mc W$ such that
	    \[\Phi_w(\rho\mt_{V|w})=\sum_k E\mt_{w,k}\rho\mt_{V|w}E_{w,k}\ct,\qquad  \sum_k E_{w,k}\ct E\mt_{w,k}=\id_V.\]
	    We can write
	    \begin{align*}
	        \Phi(\rho)\mt_{X|v,w} = \Phi'(\rho)\mt_{X|v,w} &=
	        (\id_X\otimes \bra{v,w})\Phi'(\rho)(\id_X\otimes \ket{v,w}) \\
	        &= \sum_{w'\in\mc W}\sum_k (\id_X\otimes \bra{v,w}E_{w',k})\rho_{XV|w'}(\id_X\otimes E_{w',k}\ct\ket{v,w}) \\
	        &= \sum_{w'\in\mc W}\sum_k\lN{E_{w',k}\ct\ket{v,w}}_2 \cdot \rho_{X|v',w'}
	    \end{align*}
	    where in each term of the summation, $\ket{v'}\sim E_{w',k}\ct\ket{v,w}$. Similar to the arguments in \cref{claim:lotp}, $P^{\Phi(\rho)}_{X|v,w}$ is a convex combination of $P^{\rho}_{X|v',w'}$.
	\end{proof}
	
	\subsection{Branching Program with Hybrid Memory}
	For a learning problem that corresponds to the matrix $M$, 
	a branching program of hybrid memory with $m$-bit classical memory, $q$-qubit quantum memory and length $T$ is specified as follows.
	
	At each stage $0\leq t\leq T$, the memory state of the branching program is described as a classical-quantum system $\rhot_{VW}$ over quantum memory space $\mc V=(\mathbb{C}^2)^{\otimes q}$ and classical memory space $\mc W=\{0,1\}^m$. The memory state evolves based on the samples that the branching program receives, and therefore depends on the unknown element $x\in\mt_R\mc X$.
	We can then interpret the overall systems over $XVW$, in which ${X}$ consists of an unknown concept $x$, resulting in a classical-quantum system $\rhot_{XVW}$. It always holds that the distribution of $x$ is uniform, i.e.,
	\[\rhot_X=\Tr\mt_{VW}[\rhot_{XVW}]=\frac{1}{2^n}\id_X.\]
	Initially the memory $VW$ is independent of $X$ and can be arbitrarily initialized. We assume that
	\[\rho^{(0)}_{XVW}=\frac{1}{2^n}\id_X\otimes\frac{1}{2^q}\id_V\otimes\frac{1}{2^m}\id_W.\]
	At each stage $0\leq t<T$, the branching program receives a sample $(a,b)$, where $a\in\mt_R \mc A$ and $b=M(a,x)$, and applies an operation $\Phi_{t,a,b}\in\mathscr{T}\mt_{\mc V\otimes\mc W}$ over its memory state. Thus the evolution of the entire system can be written as
	\[\rho^{(t+1)}_{XVW}=\E_{a\in\mt_R\mc A}
	\left[\sum_{x\in\mc X}\ket{x}\bra{x}\otimes\Phi_{t,a,M(a,x)}\big(\rhot_{VW|x}\big)\right].\]
	Finally, at stage $t=T$, a measurement over the computational bases is applied on $\rho^{(T)}_{VW}$, and the branching program outputs an element $\widetilde{x}\in X$ as a function of the measurement result $(v,w)\in\{0,1\}^{q+m}$. The success probability of the program is the probability that $\widetilde{x}=x$ which can be formulated as
	\[\sum_{\substack{x\in\mc X, v\in\{0,1\}^q, w\in\mc W \\ \widetilde{x}(v,w)=x}}\bra{x,v,w}\rho^{(T)}_{XVW}\ket{x,v,w}.\]
	
	\section{Main Result}
	
	\begin{theorem}\label{thm:main}
		Let $\mc X,\mc A$ be two finite sets with $n=\log_2|\mc X|$. Let $M:\mc A\times\mc X\rightarrow\{-1,1\}$ be a matrix which is a $(k',\ell')$-$L_2$ extractor with error $2^{-r'}$ for sufficiently large $k',\ell'$ and $r'$, where $\ell'\leq n$. Let
		\[r = \min\left\{\frac{1}{4}r', \frac{1}{26}\ell' + \frac{1}{6}, \frac{1}{2}(k'-1)\right\}.\]
		Let $\rho$ be a branching program for the learning problem corresponding to $M$, described by classical-quantum systems $\rhot_{XVW}$, with $q$-qubit quantum memory $V$, $m$-bit classical memory $W$ and length $T$. If $m\leq \frac{1}{44}(k'-1)\ell'$, $q\leq r-7$ and $T\leq 2^{r-2}$, the success probability of $\rho$ is at most $O(2^{q-r})$. 
	\end{theorem}
	From now on we let $k=k'-1$ and $\ell=\frac{1}{5}(\ell'-13r-2)$. Then we have the following inequalities to be used later:
	\begin{align}
	    q+r+1-r' &\ \leq \   -2 r. \label{eq:qr} \\ 
	    2\ell+9r-n &\ \leq \ -r.  \label{eq:lr} \\
	    (k-r)\ell &\ \geq \  2m+4r+1. \label{eq:m}
	\end{align}
	We leave a detailed calculation for the above inequalities in \cref{sec:parameter}.
	
	\subsection{Truncated Classical-Quantum Systems}\label{sec:trunc}
	Here we describe how to truncate a partial classical-quantum system $\rho\mt_{XVW}$ according to some property $G(v,w)$ of desire on $\rho\mt_{X|v,w}$. The goal is to remove the parts of $\rho\mt_{XVW}$ where $G$ is not satisfied. We execute the following procedure:
	\begin{enumerate}
		\item Maintain a partial system $\rho'_{XVW}$ initialized as $\rho\mt_{XVW}$, and subspaces $\mc V_w\subseteq\mc V$ initialized as $\mc V$ for each $w\in\mc W$.
		\item Pick $w\in\mc W$ and $\ket{v}\in\mc V_w$ such that $\Tr[\rho'_{X|v,w}]>0$ and $G(v,w)$ is false.
		\item Change the partial system $\rho'_{XVW}$ into the following system by projection:
		\[\big(\id_X\otimes(\id_{VW}-\ket{v,w}\bra{v,w})\big)\rho'_{XVW}\big(\id_X\otimes(\id_{VW}-\ket{v,w}\bra{v,w})\big),\]
		and change $\mc V_w$ to its subspace orthogonal to $\ket{v}$, that is
		\[\{\ket{v'}\in\mc V_w\mid \braket{v}{v'}=0\}.\]
		\item Repeat from step 2 until there is no such $w$ and $\ket{v}$. Denote the final system as $\rho^{|G}_{XVW}$.
	\end{enumerate}
	In step 2 we pick $w$ and $\ket{v}$ arbitrarily as long as it satisfies the requirements, however we could always think of it as iterating over $w\in\mc W$ and processing each $\rho_{XV|w}$ separately. The choices of $\ket{v}$ for each $w$ do affect the final system $\rho^{|G}_{XVW}$; Yet as we will see later, these choices are irrelevant to our proof. 
	
	Below, we give two useful lemmas on truncated systems.
	
	\begin{lemma}\label{lemma:id_dist}
	    For every $\ket{v}\in\mc V$ and $w\in\mc W$ such that $\Tr[\rho^{|G}_{X|v,w}]>0$, there exists $\ket{v'}$ in the remaining subspace $\mc V_w$ such that 
	    \[P^{\rho^{|G}}_{X|v,w}=P^{\rho}_{X|v',w}=P^{\rho^{|G}}_{X|v',w}.\]
	\end{lemma}

	\begin{proof}
	    It suffices to prove the lemma with one round of the truncation procedure executed. Suppose  the $\ket{v_1,w_1}$ is picked in step 2, resulting in the partial system
	    \[\rho'_{XVW}=\big(\id_X\otimes(\id_{VW}-\ket{v_1,w_1}\bra{v_1,w_1})\big)\rho\mt_{XVW}\big(\id_X\otimes(\id_{VW}-\ket{v_1,w_1}\bra{v_1,w_1})\big).\]
	    We can write
	    \begin{align*}
	        \rho'_{X|v,w} &= \big(\id_X\otimes\bra{v,w}\big)\rho'_{XVW}\big(\id_X\otimes\ket{v,w}\big) \\ &= \big(\id_X\otimes(\bra{v,w}-\braket{v,w}{v_1,w_1}\bra{v_1,w_1})\big)\rho\mt_{XVW}\big(\id_X\otimes(\ket{v,w}-\ket{v_1,w_1}\braket{v_1,w_1}{v,w})\big).
	    \end{align*}
	    \begin{itemize}
	    \item If $w\neq w_1$, then 
	    \[\rho'_{X|v,w} = \big(\id_X\otimes\bra{v,w}\big)\rho\mt_{XVW}\big(\id_X\otimes\ket{v,w}\big)
	    =\rho\mt_{X|v,w}.\]
	    And the lemma holds directly by choosing $\ket{v'}=\ket{v}$.
	    \item If $w=w_1$, then with $\braket{v_1,w_1}{v,w}=\braket{v_1}{v}=\lambda$, we have
	    \[\rho'_{X|v,w} = \big(\id_X\otimes(\bra{v}-\overline{\lambda}\bra{v_1})\bra{w}\big)\rho\mt_{XVW}\big(\id_X\otimes(\ket{v}-\lambda\ket{v_1})\ket{w}\big).\]
	    By the fact that $\Tr[\rho^{|G}_{X|v,w}]>0$, we must have $\ket{v}\neq \ket{v_1}$. Therefore if we let $\ket{v'}\sim \ket{v}-\lambda\ket{v_1}$, which is the normalized projection of $\ket{v}$ onto the orthogonal subspace of $\ket{v_1}$, the above equality implies that $P^{\rho'}_{X|v,w}=P^{\rho}_{X|v',w}$. Meanwhile, since $\braket{v_1}{v'}=0$ we have $\rho'_{X|v', w} = \rho\mt_{X|v', w}$, which completes the proof. \qedhere
	    \end{itemize}
	\end{proof}
	A direct corollary of the above lemma is that if $G(v,w)$ only depends on the distribution $P^{\rho}_{X|v,w}$, then $G(v,w)$ holds for every $\ket{v}\in\mc V$ and $w\in\mc W$ in the truncated system $\rho^{|G}_{XVW}$, even when $\ket{v}$ is not in the remaining subspace $\mc V_w$.
	
	Our second lemma is based on the following fact that bounds the trace distance of a partial system and its projection, whose proof can be found in the \cref{sec:FG-variant}.
	\begin{proposition}\label{prop:distance_weight}
		For every partial system $\rho$ and projection operator $\Pi$ on $\rho$, we have
		\[\lN[2]{\rho-\Pi\rho\Pi}_\Tr\leq 4 \Tr[\rho]^2 -4 \Tr[\Pi\rho]^2.\]
	\end{proposition}
	\begin{lemma}\label{lemma:traced}
	    For each $w\in W$, let $\ket{v_1},\ldots,\ket{v_d}$ be the states picked in step 2 within $\mc V_w$. Then
	    \[\lN{\rho\mt_{XV|w}-\rho^{|G}_{XV|w}}_\Tr\leq 3\sum_{i=1}^d \sqrt{\Tr[\rho\mt_{X|v_i,w}] \Tr[\rho\mt_{XV|w}]}.\]
	\end{lemma}
	\begin{proof}
	   In \cref{prop:distance_weight}, take $\rho$ to be $\rho\mt_{XV|w}$, and $\Pi$ to be \[\id_X\otimes\prod_{i=1}^d\left(\id_{V}-\ket{v_i}\bra{v_i}\right)=
	   \id_X\otimes\left(\id_{V}-\sum_{i=1}^d\ket{v_i}\bra{v_i}\right).\]
	   Then $\Pi\rho\Pi=\rho^{|G}_{XV|w}$ and
	   $\Tr[\Pi\rho]=\Tr[\rho\mt_{XV|w}]-\sum_{i=1}^d\Tr[\rho\mt_{X|v_i,w}]$. Therefore we have
       \begin{align*}
           \lN{\rho\mt_{XV|w}-\rho^{|G}_{XV|w}}_\Tr &\leq\sqrt{4 \Tr[\rho]^2 -4 \Tr[\Pi\rho]^2} \\
           &\leq \sqrt{8(\Tr[\rho]-\Tr[\Pi\rho])\Tr[\rho]} \\
           &= \sqrt{8\sum_{i=1}^d\Tr[\rho\mt_{X|v_i,w}] \Tr[\rho\mt_{XV|w}]} \\
           &\leq 3\sum_{i=1}^d \sqrt{\Tr[\rho\mt_{X|v_i,w}] \Tr[\rho\mt_{XV|w}]}. \qedhere
       \end{align*}
	\end{proof}
	Since $\Tr[\rho\mt_{XV|w}]\leq 1$ always holds, by summing over all $w\in\mc W$ we get the following corollary:
	\begin{corollary}\label{cor:traced}
	    Let $\ket{v_1,w_1},\ldots,\ket{v_d,w_d}$ be all of the memory states picked in step 2. Then
	    \[\lN{\rho\mt_{XVW}-\rho^{|G}_{XVW}}_\Tr\leq 3\sum_{i=1}^d \sqrt{\Tr[\rho\mt_{X|v_i,w_i}]}.\]
	\end{corollary}
	
	\subsection{Truncated Branching Program}\label{sec:truncd}
	The properties that we desire for the partial system $\rho\mt_{XVW}$ consist of three parts: %specifically on the distributions $P^\rho_{X|v,w}$:
	\begin{itemize}
		\item Small $L_2$ norm: Let $G_2(v,w)$ be the property that 
		\[\lN{P^\rho_{X|v,w}}_2\leq 2^\ell\cdot 2^{-n/2}.\]
		\item Small $L_\infty$ norm: Let $G_\infty(v,w)$ be the property that \[\lN{P^\rho_{X|v,w}}_\infty\leq 2^{2\ell+9r}\cdot 2^{-n}.\]
		\item Even division: For every $a\in\mc A$, let $G_a(v,w)$ be the property that \[|\langle M_a, P^\rho_{X|v,w}\rangle|\leq 2^{-r}.\]
	\end{itemize}

	Now we define the truncated branching program, by specifying the truncated partial classical-quantum system $\taut_{XVW}$ for each stage $t$. Initially let $\tau^{(0)}_{XVW}=\rho^{(0)}_{XVW}$. For each stage $0\leq t\leq T$, the truncation consists of three ingredients (below we ignore the superscripts on $P$ for convenience):
	\begin{enumerate}
		\item Remove parts where $\lN{P_{X|v,w}}_2$ is large. That is, let $\tau^{(t,\star)}_{XVW}=\tau^{(t)|G_2}_{XVW}$.
		\item Remove parts where $\lN{P_{X|v,w}}_\infty$ is large. This is done by two steps. 
		\begin{itemize}
			\item[-] First, let $g\in\{0,1\}^{\mc X\otimes\mc W}$ be an indicator vector such that $g(x,w)=1$ if and only if 
			\[\Tr[\tau^{(t,\star)}_{X|w}]>0\textrm{ and }P^{\tau^{(t,\star)}}_{X|w}(x)\leq 2^{2\ell+5r}\cdot 2^{-n}.\]
			Let $\tau^{(t,\circ)}_{XVW}=(gg\ct\otimes\id_V)\tau^{(t,\star)}_{XVW}(gg\ct\otimes\id_V)$, where $gg\ct$ is the projection operator acting on $\mc X\otimes\mc W$.
			\item[-] To make sure that the distributions did not change a lot after the projection $gg\ct$, for each $0\leq t< T$, let $G_t(v,w)$ be the property that
			\[\Tr[\tau^{(t,\circ)}_{X|v,w}]\geq (1-2^{-r})\Tr[\tau^{(t,\star)}_{X|v,w}].\]
			Let $\tau^{(t,\infty)}_{XVW}=\tau^{(t,\circ)|G_\infty\wedge G_t}_{XVW}$. 
		\end{itemize}
		\item For each $a\in\mc A$, remove (only for this $a$) parts where $P_{X|v,w}$ is not evenly divided by $a$. That is, for each $a\in\mc A$, let $\tau^{(t,a)}_{XVW}=\tau^{(t,\infty)|G_a}_{XVW}$.
	\end{enumerate}
	Then, if $t<T$, for each $a\in\mt_R \mc A$ we evolve the system by applying the sample operations $\Phi_{t,a,b}$ as the original branching program on $\tau^{(t,a)}_{XVW}$, so that we have
	\[\tau^{(t+1)}_{XVW}=\E_{a\in\mt_R\mc A}
	\left[\sum_{x\in\mc X}\ket{x}\bra{x}\otimes\Phi_{t,a,M(a,x)}\big(\tau^{(t,a)}_{VW|x}\big)\right].\]
	
	\subsection{Bounding the Truncation Difference}
	
	In order to show that the success probability of the original branching program $\rhot$ is low, the plan is to prove an upper bound on the success probability of the truncated branching program $\taut$, and bound the difference between the two probabilities.
	
	Here we bound the difference by the trace distance between the two systems $\rhot_{XVW}$ and $\taut_{XVW}$. We will show that the contribution to the trace distance from each one of the truncation ingredients is small, and in addition the evolution preserves the trace distance.
	
	\subsubsection{\texorpdfstring{Truncation by $G_2$}{Truncation by G2}}
	
	\begin{lemma}\label{lemma:main}
		For every $0\leq t\leq T$, $\ket{v}\in\mc V$ and $w\in\mc W$ such that $G_2(v,w)$ is violated (that is, $\lN{P^{\taut}_{X|v,w}}_2> 2^\ell\cdot 2^{-n/2}$), we must have $\Tr[\taut_{X|v,w}]<2^{-2 m}\cdot 2^{-4 r}$.
	\end{lemma}
    The lemma says, for any direction $\ket{v, w}$ picked by the truncation procedure, the weight will be small and the truncation will not change the state significantly.
	\begin{proof}
		This is our main technical lemma and we defer the proof to \cref{sect:proof_lemma_main}.
	\end{proof}
	
	Since there are at most $2^{q+m}$ such directions picked in the truncation procedure, we conclude the following corollary. 
	\begin{corollary}\label{cor:g_2}
		For every $0\leq t\leq T$, we have $\lN{\tau^{(t,\star)}_{XVW}-\taut_{XVW}}_\Tr\leq 3\cdot2^{q-2r}$.
	\end{corollary}
	\begin{proof}
	    Recall that $\tau^{(t,\star)}_{XVW}= \tau^{(t)|G_2}_{XVW}$. Since $\dim (\mc V\otimes \mc W)=2^{q+m}$, the truncation lasts for at most $2^{q+m}$ rounds. Since in every round the picked $\ket{v,w}$ has $\Tr[\taut_{X|v,w}]<2^{-2 m}\cdot 2^{-4 r}$, by \cref{cor:traced} we have 
	    \[\lN{\tau^{(t,\star)}_{XVW}-\taut_{XVW}}_\Tr\leq  3\cdot 2^{q+m}\cdot \sqrt{2^{-2 m}\cdot 2^{-4 r}}= 3\cdot2^{q-2r}. \qedhere\]
	\end{proof}
	
	\subsubsection{\texorpdfstring{Truncation by $G_\infty$}{Truncation by Ginfty}}
	
	\begin{lemma}\label{lemma:cinfty}
		For every $0\leq t\leq T$ and $w\in\mc W$ we have
		\[\sum_{\substack{x\in\mc X\\g(x,w)=0}}P^{\tau^{(t,\star)}}_{X|w}(x)\leq 2^{-5r}.\]
	\end{lemma}
	\begin{proof}
		By \cref{claim:lotp}, $P^{\tau^{(t,\star)}}_{X|w}$ is a convex combination of $P^{\tau^{(t,\star)}}_{X|v,w}$. From \cref{lemma:id_dist} we know that $G_2(P^{\tau^{(t,\star)}}_{X|v,w})$ holds for every $\ket{v}$ and $w$, and thus by convexity of $\ell_2$-norms we know that $G_2(P^{\tau^{(t,\star)}}_{X|w})$ also holds. That means
		\[\E_{x\sim P^{\tau^{(t,\star)}}_{X|w}}\left[P^{\tau^{(t,\star)}}_{X|w}(x)\right]
		=\lN[2]{P^{\tau^{(t,\star)}}_{X|w}}_2\leq 2^{2\ell}\cdot 2^{-n}. \]
		Therefore, by Markov's inequality we have
		\[\sum_{\substack{x\in\mc X\\g(x,w)=0}}P^{\tau^{(t,\star)}}_{X|w}(x)=
		\Pr_{x\sim P^{\tau^{(t,\star)}}_{X|w}}\left[P^{\tau^{(t,\star)}}_{X|w}(x)
		>2^{2\ell+5r}\cdot 2^{-n}\right]
		\leq 2^{-5r}. \qedhere\]
	\end{proof}
	\begin{corollary}\label{cor:g_inf1}
		For every $0\leq t\leq T$ and every $w \in \mc W$, we have $\tau^{(t,\circ)}_{XV|w}\leq \tau^{(t,\star)}_{XV|w}$, and 
		\[\Tr[\tau^{(t,\circ)}_{XV|w}]\geq (1-2^{-5r})\cdot \Tr[\tau^{(t,\star)}_{XV|w}].\]
		Moreover, we have $\lN{\tau^{(t,\circ)}_{XVW}-\tau^{(t,\star)}_{XVW}}_\Tr\leq 2^{-5r}$.
	\end{corollary}
	\begin{proof}
	    Since $X$ and $W$ are both classical and $\tau^{(t,\circ)}_{XVW}=(gg\ct\otimes\id_V)\tau^{(t,\star)}_{XVW}(gg\ct\otimes\id_V)$, we have
	    \[\tau^{(t,\star)}_{XV|w}-\tau^{(t,\circ)}_{XV|w} = \sum_{\substack{x\in\mc X\\g(x,w)=0}}\ket{x}\bra{x}\otimes \tau^{(t,\star)}_{V|x,w},\]
	    which is positive semi-definite. Recalling that (\Cref{eq:induced_prob}) 
	    \[\Tr[\tau^{(t,\star)}_{V|x, w}] = \bra{x,w}\tau^{(t,\star)}_{XW}\ket{x,w} = \diag \tau^{(t,\star)}_{X|w}(x) = P^{\tau^{(t, \star)}}_{X|w}(x) \Tr[\tau^{(t, \star)}_{X|w}],\] we have
	    \[
	        \Tr[\tau^{(t,\star)}_{XV|w}]-\Tr[\tau^{(t,\circ)}_{XV|w}]
	        = \sum_{\substack{x\in\mc X\\g(x,w)=0}}
	        \Tr[\tau^{(t,\star)}_{V|x,w}]
	        = \sum_{\substack{x\in\mc X\\g(x,w)=0}}
	        P^{\tau^{(t,\star)}}_{X|w}(x)\cdot \Tr[\tau^{(t,\star)}_{X|w}]
	        \leq 2^{-5r}\cdot \Tr[\tau^{(t,\star)}_{X|w}].
	    \]
	    And therefore, as $\tau^{(t,\circ)}_{XVW}-\tau^{(t,\star)}_{XVW}$ is positive semi-definite, we have
	    \[\lN{\tau^{(t,\circ)}_{XVW}-\tau^{(t,\star)}_{XVW}}_\Tr=
	    \sum_{w\in\mc W}\Tr[\tau^{(t,\star)}_{XV|w}]-\Tr[\tau^{(t,\circ)}_{XV|w}]\leq 
	    2^{-5r}\sum_{w\in\mc W}\Tr[\tau^{(t,\star)}_{X|w}]\leq 2^{-5r}. \qedhere\]
	\end{proof}
	
	\begin{lemma}
		For every $0\leq t\leq T$, $\ket{v}\in\mc V$ and $w\in\mc W$ such that $G_\infty(v,w)$ is violated (that is, $\lN{P^{\tau^{(t,\circ)}}_{X|v,w}}_\infty> 2^{2\ell+9r}\cdot 2^{-n}$) or $G_t(v,w)$ is violated (that is, $\Tr[\tau^{(t,\circ)}_{X|v,w}]< (1-2^{-r})\Tr[\tau^{(t,\star)}_{X|v,w}]$), we must have $\Tr[\tau^{(t,\circ)}_{X|v,w}]<2\cdot 2^{-4r}\cdot \Tr[\tau^{(t,\circ)}_{X|w}]$.
	\end{lemma}
	\begin{proof}
		If $G_\infty(v,w)$ is violated, let $x\in\mc X$ be the one such that $P^{\tau^{(t,\circ)}}_{X|v,w}(x)> 2^{2\ell+9r}\cdot 2^{-n}$. If $g(x,w)=0$ then $P^{\tau^{(t,\circ)}}_{X|w}(x)=0$, while if $g(x,w)=1$ then by \cref{cor:g_inf1},
		\[P^{\tau^{(t,\circ)}}_{X|w}(x)\leq  \frac{\Tr[\tau^{(t,\star)}_{X|w}]}{\Tr[\tau^{(t,\circ)}_{X|w}]}\cdot 2^{2\ell+5r}\cdot 2^{-n}\leq (1-2^{-5r})^{-1}\cdot 2^{2\ell+5r}\cdot 2^{-n}.\]
		Hence we always have
		\[\Tr[\tau^{(t,\circ)}_{X|v,w}]
		\leq \frac{P^{\tau^{(t,\circ)}}_{X|w}(x)}{P^{\tau^{(t,\circ)}}_{X|v,w}(x)}\cdot \Tr[\tau^{(t,\circ)}_{X|w}]
		\leq 2\cdot 2^{-4r}\cdot \Tr[\tau^{(t,\circ)}_{X|w}],\]
		where the first inequality comes from the fact that $\tau^{(t,\circ)}_{X|w}\geq \tau^{(t,\circ)}_{X|v,w}$ and \Cref{eq:induced_prob}. 
		
		If $G_t(v,w)$ is violated, since we know from \cref{cor:g_inf1} that
		\begin{align*}
		   \left|\Tr[\tau^{(t,\circ)}_{X|v,w}]-\Tr[\tau^{(t,\star)}_{X|v,w}]\right|
		   &\leq 
		   \lN{\tau^{(t,\circ)}_{XV|w}-\tau^{(t,\star)}_{XV|w}}_\Tr\leq 2^{-5r}\cdot \Tr[\tau^{(t,\star)}_{XV|w}] \\
		   &\leq 2^{-5r}\cdot (1-2^{-5r})^{-1}\cdot \Tr[\tau^{(t,\circ)}_{XV|w}],
		\end{align*}
		therefore from $\Tr[\tau^{(t,\circ)}_{X|v,w}]< (1-2^{-r})\Tr[\tau^{(t,\star)}_{X|v,w}]$ we deduce that
		\begin{align*}
		    \Tr[\tau^{(t,\circ)}_{X|v,w}] &< (2^r-1)\cdot\left( \Tr[\tau^{(t,\star)}_{X|v,w}]-\Tr[\tau^{(t,\circ)}_{X|v,w}]\right)  \\
		    & \leq (2^r-1)\cdot 2^{-5r}\cdot (1-2^{-5r})^{-1}\cdot \Tr[\tau^{(t,\circ)}_{XV|w}] \\
		    & < 2\cdot 2^{-4r}\cdot \Tr[\tau^{(t,\circ)}_{X|w}]. 		\qedhere
		\end{align*}

	\end{proof}

	\begin{corollary}\label{cor:g_inf2}
		For every $0\leq t\leq T$, we have $\lN{\tau^{(t,\infty)}_{XVW}-\tau^{(t,\circ)}_{XVW}}_\Tr\leq 5\cdot 2^{q-2r}$.
	\end{corollary}
	\begin{proof}
	    Recall that $\tau^{(t,\infty)}_{XVW}= \tau^{(t,\circ)|G_\infty\wedge G_t}_{XVW}$. For each $w\in\mc W$, the truncation picks at most $\dim\mc V=2^q$ states $\ket{v,w}$, each with $\Tr[\tau^{(t,\circ)}_{X|v,w}]<2\cdot 2^{-4r}\cdot \Tr[\tau^{(t,\circ)}_{X|w}]$. Therefore by applying \cref{lemma:traced} for each $w\in\mc W$, we have
	    \[\lN{\tau^{(t,\infty)}_{XVW}-\tau^{(t,\circ)}_{XVW}}_\Tr \leq
	    3\cdot\sum_{w\in\mc W}2^q\cdot \sqrt{2\cdot 2^{-4r}}\cdot \Tr[\tau^{(t,\circ)}_{X|w}]\leq 5\cdot 2^{q-2r}. \qedhere\]
	\end{proof}
	
	\subsubsection{\texorpdfstring{Truncation by $G_a$}{Truncation by Ga}}
	
	Notice that in the truncation step from $\tau^{(t,\star)}$ to $\tau^{(t,\circ)}$, the distribution $P^{\tau^{(t,\star)}}_{X|v,w}$ might change and not satisfy $G_2$ anymore. However, with the truncation by $G_t$, any such distribution that changes too much is eliminated, and we have the following guarantee.
	
	\begin{lemma}\label{lemma:2norm}
	    For every $0\leq t\leq T$, $\ket{v}\in\mc V$ and $w\in\mc W$, we have
	    \[\lN{P^{\tau^{(t,\infty)}}_{X|v,w}}_2\leq (1-2^{-r})^{-1}\cdot 2^\ell\cdot 2^{-n/2}.\]
	\end{lemma}
	\begin{proof}
	    By \cref{lemma:id_dist}, there exists $\ket{v'}\in\mc V$ such that $P^{\tau^{(t,\infty)}}_{X|v,w}=P^{\tau^{(t,\infty)}}_{X|v',w}=P^{\tau^{(t,\circ)}}_{X|v',w}$. The truncation by $G_t$ ensures that $\Tr[\tau^{(t,\circ)}_{X|v',w}]\geq (1-2^{-r})\Tr[\tau^{(t,\star)}_{X|v',w}]$, and therefore
	    \[\lN{P^{\tau^{(t,\infty)}}_{X|v,w}}_2=\lN{P^{\tau^{(t,\circ)}}_{X|v',w}}_2
	    =\frac{\lN{\diag\tau^{(t,\circ)}_{X|v',w}}_2}{\Tr[\tau^{(t,\circ)}_{X|v',w}]}
	    \leq \frac{\lN{\diag\tau^{(t,\star)}_{X|v',w}}_2}{(1-2^{-r})\Tr[\tau^{(t,\star)}_{X|v',w}]}
	    \leq (1-2^{-r})^{-1}\cdot 2^\ell\cdot 2^{-n/2}. \qedhere
	    \]
	\end{proof}
	
	\begin{lemma}\label{lemma:g_a}
		For every partial classical-quantum system $\tau\mt_{XV}$ over $\mc X\otimes\mc V$ such that $\lN{P^{\tau}_{X|v}}_2\leq 2^{\ell'}\cdot2^{-n/2}$ holds for every $\ket{v}\in\mc V$, we have
		\[\Pr_{a\in\mt_R \mc A}\left[\exists \ket{v}\in\mc V, |\langle M_a, P^{\tau}_{X|v}\rangle|\geq 2^{-r}\right]\leq 2^{-2r}.\]
	\end{lemma}
	\begin{proof} 
		Notice that we can think of $\tau\mt_V=\Tr\mt_X[\tau\mt_{XV}]$ to be $\id_V$. This is because we can first assume that $\tau\mt_V$ is full rank (otherwise change $\mc V$ to its subspace and the conclusion in this lemma still holds), and if we have diagonalization $Q\ct\tau\mt_V Q=\id_V$ for some non-singular $Q$, then consider the new system
		\[\tau'_{XV}=(\id_X\otimes Q\ct)\tau\mt_{XV}(\id_X\otimes Q),\]
		and the set of distributions $\{P^{\tau}_{X|v}\}$ and $\{P^{\tau'}_{X|v}\}$ over $\ket{v}\in\mc V$ are the same, since $P^{\tau'}_{X|v}=P^{\tau}_{X|v'}$ for $\ket{v'}\sim Q\ket{v}$.  With $\tau\mt_V=\id_V$ we have $\Tr[\tau\mt_{X|v}]=1$ for every $\ket{v}\in\mc V$, and thus $P^\tau_{X|v}=\diag \tau\mt_{X|v}$.
		
		Let $\mc A'\subseteq\mc A$ be the set of $a\in\mc A$ such that there exists $\ket{v}\in\mc V$ with $|\langle M_a, P^{\tau}_{X|v}\rangle|\geq 2^{-r}$. For each $a\in\mc A'$, let
		\[\sigma_a=\Tr\mt_X[(\Diag M_a\otimes\id_V)\tau\mt_{XV}]\]
		which is a Hermitian operator on $\mc V$. There exists $\ket{v}\in\mc V$ such that \[|\bra{v}\sigma_a\ket{v}|=|\langle M_a, \diag \tau\mt_{X|v}\rangle|=|\langle M_a, P^{\tau}_{X|v}\rangle|\geq 2^{-r},\]
		which means that $\|\sigma_a\|_2\geq 2^{-r}$. Now let $\ket{u}$ be a uniformly random unit vector in $\mc V$, and by \cref{lemma:anticc} we know that for some absolute constant $c$,
		\[\Pr_{\ket u}\left[|\bra{u}\sigma_a\ket{u}|\geq 2^{-r'}\right]\geq 1-2^{(q+r-r')/2}c-e^{-2^q}\geq 1-2^{-r}c-e^{-1}\geq 1/2.\]
		The second last inequality comes from \cref{eq:qr}, while the last inequality is because of the assumption that $r$ is sufficiently large.
		
		Since the above holds for every $a \in \mc{A}'$, it implies that $\Pr_{a \in \mc{A}', \ket u}[|\langle u|\sigma_a| u\rangle| \geq 2^{-r'}]$ is at least $1/2$. 
		It means that there exists some $\ket{u}\in\mc V$ such that $|\bra{u}\sigma_a\ket{u}|\geq 2^{-r'}$ for at least half of $a\in\mc A'$. On the other hand, since $M$ is a $(k',\ell')$-extractor with error $2^{-r'}$, and $\lN{P^{\tau}_{X|u}}_2\leq 2^{\ell'}\cdot2^{-n/2}$, there are at most $2^{-k'}$ fraction of $a\in\mc A$ such that $|\bra{u}\sigma_a\ket{u}|=|\langle M_a, P^{\tau}_{X|u}\rangle|\geq 2^{-r'}$. That means
		\[\Pr_{a\in\mt_R \mc A}\left[a\in\mc A'\right]\leq 2\cdot 2^{-k'}\leq 2^{-2r}. \]
		Here $k' - 1 \geq 2 r$, by the definition of $r$. 
	\end{proof}
	
	\begin{corollary}\label{cor:g_a}
		For every $0\leq t\leq T$, we have $\E_{a\in\mt_R \mc A}\lN{\tau^{(t,a)}_{XVW}-\tau^{(t,\infty)}_{XVW}}_\Tr\leq 2^{-2r}$.
	\end{corollary}
	\begin{proof}
	    For each $w\in\mc W$, the partial system $\tau^{(t,\infty)}_{XV|w}$ satisfies the condition of \cref{lemma:g_a} since for every $\ket{v}\in\mc V$,
	    \[\lN{P^{\tau^{(t,\infty)}}_{X|v,w}}_2\leq (1-2^{-r})^{-1}\cdot 2^\ell\cdot 2^{-n/2}
	    \leq 2^{\ell'}\cdot 2^{-n/2}.\]
	    Notice that for each $a\in\mc A$ such that there does not exist $\ket{v}\in\mc V$ with $|\langle M_a, P^{\tau^{(t,\infty)}}_{X|v,w}\rangle|\geq 2^{-r}$ (that is, when $G_a(v,w)$ holds for every $\ket{v}\in\mc V$), the sub system $\tau^{(t,\infty)}_{XV|w}$ is not touched in the truncation by $G_a$ and we have $\tau^{(t,a)}_{XV|w}=\tau^{(t,\infty)}_{XV|w}$. Therefore 
	    \begin{align*}
	        \E_{a\in\mt_R \mc A}\lN{\tau^{(t,a)}_{XVW}-\tau^{(t,\infty)}_{XVW}}_\Tr &=
	        \sum_{w\in\mc W} \E_{a\in\mt_R \mc A}\lN{\tau^{(t,a)}_{XV|w}-\tau^{(t,\infty)}_{XV|w}}_\Tr \\
	        & \leq \sum_{w\in\mc W}\Pr_{a\in\mt_R \mc A}\left[\exists \ket{v}\in\mc V, |\langle M_a, P^{\tau^{(t,\infty)}}_{X|v,w}\rangle|\geq 2^{-r}\right]\cdot  \Tr[\tau^{(t,\infty)}_{XV|w}] \\
	        &\leq 2^{-2r}\cdot \sum_{w\in\mc W}\Tr[\tau^{(t,\infty)}_{XV|w}] 
	        \leq 2^{-2r}.
	    \qedhere
	    \end{align*}
	\end{proof}

	\subsubsection{Evolution preserves trace distance}
	
	\begin{lemma}\label{lemma:evol}
		For every $0\leq t<T$, we have $\lN{\tau^{(t+1)}_{XVW}-\rho^{(t+1)}_{XVW}}_\Tr\leq\E_{a\in\mt_R \mc A}\lN{\tau^{(t,a)}_{XVW}-\rho^{(t)}_{XVW}}_\Tr$.
	\end{lemma}
	\begin{proof}
		Recall that
		\[\rho^{(t+1)}_{XVW}=\E_{a\in\mt_R\mc A}
		\left[\sum_{x\in\mc X}\ket{x}\bra{x}\otimes\Phi_{t,a,M(a,x)}\big(\rhot_{VW|x}\big)\right],\]
		\[\tau^{(t+1)}_{XVW}=\E_{a\in\mt_R\mc A}
		\left[\sum_{x\in\mc X}\ket{x}\bra{x}\otimes\Phi_{t,a,M(a,x)}\big(\tau^{(t,a)}_{VW|x}\big)\right].\]
		Therefore by triangle inequality and contractivity of quantum channels under trace norms, 
		\begin{align*}
			\lN{\tau^{(t+1)}_{XVW}-\rho^{(t+1)}_{XVW}}_\Tr &\leq  \E_{a\in\mt_R\mc A}
			\left\|\sum_{x\in\mc X}\ket{x}\bra{x}\otimes\left(\Phi_{t,a,M(a,x)}\big(\tau^{(t,a)}_{VW|x}\big)-
			\Phi_{t,a,M(a,x)}\big(\rhot_{VW|x}\big)\right)\right\|_\Tr \\
			&=\E_{a\in\mt_R\mc A}\sum_{x\in\mc X}\left\|\Phi_{t,a,M(a,x)}\big(\tau^{(t,a)}_{VW|x}\big)-
			\Phi_{t,a,M(a,x)}\big(\rhot_{VW|x}\big)\right\|_\Tr \\
			&\leq \E_{a\in\mt_R\mc A}\sum_{x\in\mc X}\lN{\tau^{(t,a)}_{VW|x}-\rhot_{VW|x}}_\Tr \\
			&= \E_{a\in\mt_R \mc A}\lN{\tau^{(t,a)}_{XVW}-\rho^{(t)}_{XVW}}_\Tr. \qedhere
		\end{align*}
	\end{proof}
	
	\subsection{\texorpdfstring{Proof of \cref{thm:main}}{Proof of the Main Theorem}}
	
	\begin{proof}
		First, combining \cref{cor:g_2,cor:g_inf1,cor:g_inf2,cor:g_a,lemma:evol} we have
		\[\lN{\tau^{(t+1)}_{XVW}-\rho^{(t+1)}_{XVW}}_\Tr\leq \lN{\taut_{XVW}-\rhot_{XVW}}_\Tr+8\cdot 2^{q-2r} +2^{-5r}+ 2^{-2r}.\]
		Since $\tau^{(0)}_{XVW}=\rho^{(0)}_{XVW}$, by triangle inequality we know that  $\lN{\tau^{(T)}_{XVW}-\rho^{(T)}_{XVW}}_\Tr\leq T\cdot 10\cdot 2^{q-2r}\leq 10\cdot 2^{q-r}$, and thus
		\[\lN{\tau^{(T,\infty)}_{XVW}-\rho^{(T)}_{XVW}}_\Tr\leq 10\cdot 2^{q-r}+8\cdot 2^{q-2r} + 2^{-5 r}.\] 
		This bounds the difference between the measurement probabilities of $\tau^{(T,\infty)}_{XVW}$ and $\rho^{(T)}_{XVW}$ under any measurement, specifically the difference between the success probability of the branching program $\rho$ and the following value on $\tau$:
		\[\sum_{\substack{x\in\mc X,v\in\{0,1\}^q,w\in\mc W\\ \widetilde{x}(v,w)=x}}
		\bra{x,v,w}\tau^{(T,\infty)}_{XVW}\ket{x,v,w}
		=\sum_{v\in\{0,1\}^q,w\in\mc W}\Tr[\tau^{(T,\infty)}_{X|v,w}]\cdot P^{\tau^{(T,\infty)}}_{X|v,w}(\widetilde{x}(v,w)).\]
		Since $\lN{P^{\tau^{(T,\infty)}}_{X|v,w}}_\infty\leq 2^{2\ell+9r}\cdot 2^{-n}$ and $\Tr[\tau^{(T,\infty)}_{XVW}]\leq 1$, the above value is at most $2^{2\ell+9r}\cdot 2^{-n}$. Therefore the success probability of the branching program $\rho$ is at most (recall that $2\ell+9r-n\leq -r$) 
		\[10\cdot 2^{q-r}+8\cdot 2^{q-2r} + 2^{-5 r}+2^{2\ell+9r}\cdot 2^{-n}=O(2^{q-r}). \qedhere\]
	\end{proof}

	\section{\texorpdfstring{Proof of \cref{lemma:main}}{Proof of Main Lemma}}\label{sect:proof_lemma_main}
	
	The first step towards proving \cref{lemma:main} is to analyze how $P^{\tau^{(t)}}_{X|v,w}$ evolves according to the rule
	\[\tau^{(t+1)}_{XVW}=\E_{a\in\mt_R\mc A}
	\left[\sum_{x\in\mc X}\ket{x}\bra{x}\otimes\Phi_{t,a,M(a,x)}\big(\tau^{(t,a)}_{VW|x}\big)\right].\]
	We introduce the following notations. For every $a\in\mc A$ and $b\in\{-1,1\}$, let
	\[\mathbbm{1}_{a,b}=\frac{1}{2}(\vec{1}+b\cdot M_a),\]
	which is a $0$-$1$ vector that indicates whether $M(a,x)=b$. Let 
	\begin{equation}\label{eq:tauab}
	    \tau^{(t,a,b)}_{XVW}=(\Diag \mathbbm{1}_{a,b}\otimes\id_{VW})\tau^{(t,a)}_{XVW},
	\end{equation}
	so that we can write
	\begin{equation}\label{eq:evolution}
	    \tau^{(t+1)}_{XVW}=\E_{a\in\mt_R\mc A}
	\left[(\id_X\otimes\Phi_{t,a,1})\big(\tau^{(t,a,1)}_{XVW}\big)+(\id_X\otimes\Phi_{t,a,-1})\big(\tau^{(t,a,-1)}_{XVW}\big)\right].
	\end{equation}
	Thus \cref{claim:chan} implies that $P^{\tau^{(t+1)}}_{X|v,w}$ is a convex combination of $P^{\tau^{(t,a,b)}}_{X|v',w'}$ for some $a,b,w'$ and $\ket{v'}$.
	
	\subsection{Target Distribution and Badness}
	Before considering the target distribution, let us first establish that the $\ell_2$-norms of $P^{\tau^{(t)}}_{X|v,w}$ cannot be too large:
	\begin{lemma} For every $0\leq t\leq T$, $\ket v\in\mc V$, $w\in\mc W$, we have
		\[\lN{P^{\tau^{(t)}}_{X|v,w}}_2\leq 4\cdot 2^\ell\cdot 2^{-n/2}.\]
	\end{lemma}
	\begin{proof}
	    When $t=0$ the statement is clearly true as $P^{\tau^{(0)}}_{X|v,w}$ is always uniform. 
	    
	    Now assume $t>0$. By \cref{lemma:id_dist} and \cref{lemma:2norm} we know that \[\lN{P^{\tau^{(t-1,a)}}_{X|v',w'}}_2\leq (1-2^{-r})^{-1}\cdot 2^\ell\cdot 2^{-n/2}\]
	    for every $w'\in\mc W, \ket{v'}\in\mc V$ and $a\in\mc A$, as $\tau_{XVW}^{(t-1, a)}$ is truncated from $\tau_{XVW}^{(t-1, \infty)}$. Since $G_a(P^{\tau^{(t-1,a)}}_{X|v',w'})$ is true, meaning that the distribution is evenly divided by $a$, we further have
		\[\lN{P^{\tau^{(t-1,a,b)}}_{X|v',w'}}_2=
		\frac{\lN{\mathbbm{1}_{a,b}\cdot P^{\tau^{(t-1,a)}}_{X|v',w'}}_2}
		{\lN{\mathbbm{1}_{a,b}\cdot P^{\tau^{(t-1,a)}}_{X|v',w'}}_1}
		\leq 2(1-2^{-r})^{-1}\cdot \lN{P^{\tau^{(t-1,a)}}_{X|v',w'}}_2\leq 4\cdot 2^{\ell}\cdot 2^{-n/2}.\]
		Since $P^{\tau^{(t)}}_{X|v,w}$ is a convex combination of $P^{\tau^{(t-1,a,b)}}_{X|v',w'}$, by convexity its $\ell_2$-norm is bounded by $4 \cdot 2^{\ell}\cdot 2^{-n/2}$.
	\end{proof}
	
	From now on we use $P$ to denote a fixed target distribution (which we will later choose to be the distribution in \cref{lemma:main}), such that
	\[2^\ell \cdot 2^{-n/2}\leq \|P\|_2 \leq 4\cdot 2^\ell \cdot 2^{-n/2}.\]
    We want to bound the progress of $\langle P^{\tau^{(t)}}_{X|v,w}, P\rangle$, which starts off as $2^{-n}$ at $t=0$, and becomes at least $2^{2\ell}\cdot 2^{-n}$ when $P^{\tau^{(t)}}_{X|v,w}=P$. Note that by Cauchy-Schwarz we always have
    \begin{equation}\label{eq:cauchy}
        \langle P^{\tau^{(t)}}_{X|v,w}, P\rangle\leq \lN{P^{\tau^{(t)}}_{X|v,w}}_2\lN{P}_2\leq 16\cdot 2^{2\ell}\cdot 2^{-n}.
    \end{equation}
    In order to bound the progress, we introduce some new notations. For any superscript (such as $(t,a)$) on the partial systems, we use $\sigma\mt_{XVW}$ to denote $\tau\mt_{XVW}(\Diag P\otimes\id_{VW})$. Notice that
	\[\Tr[\sigma\mt_{X|v,w}]=\Tr[\tau\mt_{X|v,w}\Diag P]=\Tr[\tau\mt_{X|v,w}]\cdot \langle P^\tau_{X|v,w},P\rangle.\]
	Similarly, $P^\sigma_{X|v,w}$ can be deduced from $P^\tau_{X|v,w}$ via
	\begin{equation}\label{eq:sigma}
	    P^\sigma_{X|v,w}(x)=\frac{\Tr[\tau\mt_{X|v,w}]}{\Tr[\sigma\mt_{X|v,w}]}\cdot
	    P^\tau_{X|v,w}(x)\cdot P(x) = \frac{P^\tau_{X|v,w}(x)\cdot P(x)}{\langle P^\tau_{X|v,w},P\rangle}.
	\end{equation}
	Therefore we can bound the $\ell_2$ norm of $P^\sigma_{X|v,w}$ as
	\[
	\lN{P^\sigma_{X|v,w}}_2 \leq \frac{1}{\langle P^\tau_{X|v,w},P\rangle}
	\cdot\lN{P^\tau_{X|v,w}}_\infty\cdot \lN{P}_2.
	\]

	Now we can identity the places where $\langle P^{\tau^{(t)}}_{X|v,w}, P\rangle$ increases by a lot, which happens when the \emph{inner product} is not evenly divided by some $a\in\mc A$ (we will see the reason in the analysis later). Formally, at stage $0\leq t<T$, we say $(w,a)$ is \emph{bad} if
	\begin{equation}\label{eq:badness}
	\exists\ket{v}\in\mc V\textrm{, s.t. }
	|\langle M_a, P^{\sigma^{(t,a)}}_{X|v,w}\rangle|>2^{-r}
	\textrm{ and }
	\langle P^{\tau^{(t,a)}}_{X|v,w}, P\rangle\geq \frac{1}{2}\cdot 2^{-n}.
	\end{equation}

	\begin{lemma}\label{lemma:badness}
		For every $0\leq t<T$ and $w\in\mc W$, we have
		\[\Pr_{a\in\mt_R \mc A}[\textrm{$(w,a)$ is bad}]\leq 2^{-k}.\]
	\end{lemma}
	\begin{proof}
		Since $\tau^{(t,a)}_{XVW}$ is truncated from $\tau^{(t,\infty)}_{XVW}$,  \cref{lemma:id_dist} shows that for every $\ket{v}\in\mc V$, $w\in\mc W$ and $a\in\mc A$ there is $\ket{v'}\in\mc V$ such that
		\[P^{\tau^{(t,a)}}_{X|v,w}=P^{\tau^{(t,\infty)}}_{X|v',w}\]
		and by \cref{eq:sigma} it also implies that
		\[P^{\sigma^{(t,a)}}_{X|v,w}=P^{\sigma^{(t,\infty)}}_{X|v',w}.\]
		Now fix some $w\in\mc W$, and let $\mc A'\subseteq\mc A$ be the set of of $a\in\mc A$ such that
		\[\exists\ket{v}\in\mc V\textrm{, s.t. }
		|\langle M_a, P^{\sigma^{(t,\infty)}}_{X|v,w}\rangle|>2^{-r}
		\textrm{ and }
		\langle P^{\tau^{(t,\infty)}}_{X|v,w}, P\rangle\geq \frac{1}{2}\cdot 2^{-n}.\]
		Then $\mc A'$ contains all $a$ such that $(w,a)$ is bad, and our goal is to bound the fraction of $\mc A'$ in $\mc A$.
		
		In the rest of the proof we temporarily omit the super script and write $\tau^{(t,\infty)}$ and $\sigma^{(t,\infty)}$ simply as $\tau$ and $\sigma$. For the same reason as in \cref{lemma:g_a} we can assume that $\tau\mt_{V|w}=\id_V$, and thus 
		\[\bra{v}\sigma\mt_{V|w}\ket{v}=\Tr[\sigma\mt_{X|v,w}]=\langle P^\tau_{X|v,w},P\rangle,
		{\textrm{ and } \Tr[\sigma\mt_{XV|w}]=\langle P^\tau_{X|w},P\rangle\leq 16\cdot 2^{2\ell}\cdot 2^{-n}.}\] 
		where the last inequality is by \cref{lemma:2norm} and Cauchy-Schwarz, in the same way as \cref{eq:cauchy}.
		
		Suppose that we have diagonalization $\sigma\mt_{V|w}=U\ct DU$, where $U$ is unitary and $D$ is diagonal and non-negative. Let $\mc V'\subseteq\mc V$ be the subspace spanned by $U\ct\ket{e}$ over the computational basis vectors $\ket{e}\in\mc V$ such that $\bra{e}D\ket{e}\geq 2^{-4r}\cdot {2^{-2\ell}}\cdot 2^{-n}$. So for every $\ket{v}\in\mc V'$ we have \[\langle P^\tau_{X|v,w},P\rangle=\Tr[\sigma\mt_{X|v,w}]\geq 2^{-4r}\cdot {2^{-2\ell}}\cdot 2^{-n}.\]
		
		We claim that for every $a\in\mc A'$, there exists $\ket{v}\in\mc V'$ such that $|\langle M_a, P^\sigma_{X|v,w}\rangle|>\frac{1}{2}\cdot 2^{-r}$. To prove the claim, let $\Pi$ be the projection operator from $\mc V$ to $\mc V'$, and then $(\id_X\otimes\Pi)\sigma\mt_{XV|w}(\id_X\otimes\Pi)$ can be conceptually seen as a truncated partial system $\sigma^{|G}_{XV|w}$ where $G(v,w)$ holds when $\Tr[\sigma\mt_{X|v,w}]\geq 2^{-4r -2 \ell}\cdot 2^{-n}$ for the fixed $w$. By \cref{lemma:traced} we have
		\[\lN{\sigma^{|G}_{XV|w}-\sigma\mt_{XV|w}}_\Tr\leq 3\cdot 2^q\cdot {\sqrt{2^{-4r-2\ell-n}\cdot \Tr[\sigma\mt_{XV|w}]}}\leq 12\cdot 2^{q-2r}\cdot 2^{-n}.\]
        Since $a\in\mc A'$, assume for $\ket{u}\in\mc V$ we have $|\langle M_a, P^\sigma_{X|u,w}\rangle|>2^{-r}$ and $\Tr[\sigma\mt_{X|u,w}]=\langle P^\tau_{X|u,w}, P\rangle\geq \frac{1}{2}\cdot 2^{-n}$. Let $\ket{v}\sim\Pi\ket{u}$, then we have 
        \begin{align*}
        \lN{P^\sigma_{X|u,w}-P^\sigma_{X|v,w}}_1 
         =&\ \lN{P^\sigma_{X|u,w}-P^{\sigma^{|G}}_{X|u,w}}_1 
         \leq\left\| \frac{\sigma\mt_{X|u,w}}{\Tr[\sigma\mt_{X|u,w}]} 
         -\frac{\sigma^{|G}_{X|u,w}}{\Tr[\sigma^{|G}_{X|u,w}]}\right\|_\Tr\\
         \leq &\   \left\|{\frac{\sigma\mt_{X|u,w}}{\Tr[\sigma\mt_{X|u,w}]} -
         \frac{\sigma^{|G}_{X|u,w}}{\Tr[\sigma\mt_{X|u,w}]}}\right\|_\Tr 
         +  \left\|{\frac{\sigma^{|G}_{X|u,w}}{\Tr[\sigma\mt_{X|u,w}]} -
         \frac{\sigma^{|G}_{X|u,w}}{\Tr[\sigma^{|G}_{X|u,w}]}}\right\|_\Tr \\
         =&\   \left\|{\frac{\sigma\mt_{X|u,w}}{\Tr[\sigma\mt_{X|u,w}]} - \frac{\sigma^{|G}_{X|u,w}}{\Tr[\sigma\mt_{X|u,w}]}}\right\|_\Tr
         + \left|\frac{1}{\Tr[\sigma\mt_{X|u,w}]}-\frac{1}{\Tr[\sigma^{|G}_{X|u,w}]}\right|\cdot \Tr[\sigma^{|G}_{X|u,w}]\\
        =& \frac{\lN{\sigma\mt_{X|u,w}-\sigma^{|G}_{X|u,w}}_\Tr}{\Tr[\sigma\mt_{X|u,w}]}
        +\frac{\left|\Tr[\sigma^{|G}_{X|u,w}]-\Tr[\sigma\mt_{X|u,w}]\right|}{\Tr[\sigma\mt_{X|u,w}]}\\
        \leq&\ \frac{2\lN{\sigma\mt_{X|u,w}-\sigma^{|G}_{X|u,w}}_\Tr}{\Tr[\sigma\mt_{X|u,w}]} 
        \leq \frac{2\lN{\sigma\mt_{XV|w}-\sigma^{|G}_{XV|w}}_\Tr}{\Tr[\sigma\mt_{X|u,w}]} 
        \leq 48\cdot 2^{q-2r}\leq \frac{1}{2}\cdot 2^{-r},
        \end{align*}
		where the last step is due to $q\leq r-7$. Thus
		\[|\langle M_a, P^\sigma_{X|v,w}\rangle|\geq |\langle M_a, P^\sigma_{X|u,w}\rangle|-\lN{P^\sigma_{X|u,w}-P^\sigma_{X|v,w}}_1> \frac{1}{2}\cdot 2^{-r}.\]
		
		Similarly to the proof for \cref{lemma:g_a}, for each $a\in\mc A'$ let
		\[\pi_a=\Tr\mt_X[(\Diag M_a\otimes U\ct D^{-1/2}U)\cdot \sigma\mt_{XV|w}\cdot (\id_X\otimes U\ct D^{-1/2}U)]\]
		which is a Hermitian operator on $\mc V$. For each $\ket{v}\in\mc V$, let $\ket{v'}\sim U\ct D^{1/2}U\ket{v}$. Recall that $\sigma\mt_{V|w}=U\ct DU$, and therefore
		\begin{align*}
		    P^{\sigma}_{X|v,w} 
		    &=\frac{\diag\ (\id_X\otimes\bra{v})\sigma\mt_{XV|w}(\id_X\otimes\ket{v})}
		    {\bra{v}\sigma\mt_{V|w}\ket{v}} \\
		    &=\frac{\diag\ (\id_X\otimes\bra{v'}U\ct D^{-1/2}U)\sigma\mt_{XV|w}(\id_X\otimes U\ct D^{-1/2}U\ket{v'})}
		    {\bra{v'}U\ct D^{-1/2}U\sigma\mt_{V|w}U\ct D^{-1/2}U\ket{v'}} \\
		    &= \diag\ (\id_X\otimes\bra{v'}U\ct D^{-1/2}U)\sigma\mt_{XV|w}(\id_X\otimes U\ct D^{-1/2}U\ket{v'}).
		\end{align*}
		And that means
		\[\bra{v'}\pi_a\ket{v'}=\left\langle M_a, \diag\ (\id_X\otimes\bra{v'}U\ct D^{-1/2}U)\sigma\mt_{XV|w}(\id_X\otimes U\ct D^{-1/2}U\ket{v'})\right\rangle
		=\langle M_a, P^\sigma_{X|v,w}\rangle.\]
		We showed above that there exists $\ket{v}\in\mc V'$, and thus $\ket{v'}\in\mc V'$ such that
		\[|\bra{v'}\pi_a\ket{v'}|=\left|\langle M_a, P^\sigma_{X|v,w}\rangle\right|\geq\frac{1}{2}\cdot 2^{-r},\]
		which means that for $\Pi\pi_a\Pi$, the restriction of $\pi_a$ on $\mc V'$, we have $\|\Pi\pi_a\Pi\|_2\geq \frac{1}{2}\cdot 2^{-r}$. Now consider a uniformly random unit vector $\ket{v'}$ in $\mc V'$, and by \cref{lemma:anticc} we know that for some absolute constant $c$,
		\[\Pr_{\ket v'}\left[|\bra{v'}\sigma_a\ket{v'}|\geq 2^{-r'}\right]\geq 1-2^{(q+r+1-r')/2}c-e^{-2^q}\geq 1-2^{-r}c-e^{-1}\geq \frac{1}{2}.\]
		Therefore, for the random vector $\ket{v}\sim U\ct D^{-1/2}U\ket{v'}$ where $\ket{v'}$ is uniform in $\mc V'$, we conclude that
		\[\Pr_{\ket v}\left[|\langle M_a, P^\sigma_{X|v,w}\rangle|\geq 2^{-r'}\right]\geq \frac{1}{2}.\]
		On the other hand, as $\ket{v'}\in\mc V'$, it also holds that $\ket{v}\in\mc V'$, therefore $\langle P^\tau_{X|v,w}, P\rangle\geq 2^{-4r}\cdot 2^{-2\ell}\cdot 2^{-n}$ is always true. Thus there exists a $\ket{v}\in\mc V$ that simultaneously satisfies
		\[\langle P^\tau_{X|v,w}, P\rangle\geq 2^{-4r}\cdot 2^{-2\ell}\cdot 2^{-n} \quad\textrm{and}\quad
		|\langle M_a, P^\sigma_{X|v,w}\rangle|\geq 2^{-r'}\]
		for at least $1/2$ of $a\in\mc A'$. Since
		\[\lN{P^\sigma_{X|v,w}}_2\leq 
		\frac{1}{\langle P^\tau_{X|v,w},P\rangle}
		\cdot\lN{P^\tau_{X|v,w}}_\infty\cdot \lN{P}_2
		\leq  4\cdot 2^{{5\ell+13r}}\cdot 2^{-n/2}= 2^{\ell'}\cdot 2^{-n/2},\]
		and $M$ is a $(k',\ell')$-extractor with error $2^{-r'}$, there are at most $2^{-k'}$ fraction of $a\in\mc A$ such that $|\langle M_a, P^{\sigma}_{X|v',w}\rangle|\geq 2^{-r'}$, which means that
		\[\Pr_{a\in\mt_R \mc A}[\textrm{$(w,a)$ is bad}]\leq \Pr_{a\in\mt_R \mc A}[a\in\mc A']\leq 2\cdot 2^{-k'}=2^{-k}. \qedhere\]
	\end{proof}
	
	\subsection{Badness Levels}
	At stage $t$, for each classical memory state $w\in\mc W$ we count how many times the path to it has been bad, which is a random variable depending on the previous random choices of $a\in\mc A$. This is stored in another classical register $B$, which we call \emph{badness level} and takes values $\beta\in\{0,\ldots,T\}$. It is initially set to be $0$, that is, we let
	\[\tau^{(0)}_{XVWB}=\tau^{(0)}_{XVW}\otimes\ket{0}\bra{0}_B.\]
	We ensure that the distribution of $B$ always only depends on $W$ and is independent of $X$ and $V$ conditioned on $W$, using the following updating rules on the combined system $\tau\mt_{XVWB}$ for each stage $0\leq t< T$:
	\begin{itemize}
	    \item The truncation steps are executed independently of $B$. Therefore, for each $a\in\mc A$ we let
	    \begin{equation}\label{eq:trunc}
	        \tau^{(t,a)}_{XVWB}=\sum_{w\in\mc W}\tau^{(t,a)}_{XV|w}\otimes\ket{w}\bra{w}\otimes\Diag P^{\taut}_{B|w}.
	    \end{equation}
	    \item The value of $B$ updates before the evolution step, where for each $a\in\mc A$ and $b\in\{-1,1\}$ we let
	    \[\tau^{(t,a,b)}_{XVWB}=
	    (\Diag \mathbbm{1}_{a,b}\otimes \id_V\otimes U_a)\tau^{(t,a)}_{XVWB}(\id_{XV}\otimes U_a\ct).\]
	    Here $U_a$ is a permutation operator, depending on $\tau^{(t, a)}_{XVW}$, acting on $\mc W\otimes\{0,\ldots,T\}$ such that
    	\[U_a\ket{w}\ket{\beta}=\left\{\begin{array}{ll}
    	\ket{w}\ket{(\beta+1) \mathop{\mathrm{mod}} (T+1)} &\textrm{ if $(w,a)$ is bad,} \\
    	\ket{w}\ket{\beta} &\textrm{ otherwise.}
    	\end{array}\right.\]
    	\item For the evolution step, we apply the channels $\Phi_{t,a,b}$ on the memories $W$ and $V$ to get
    	\[\tau^{(t+1)}_{XVWB}=\E_{a\in\mt_R\mc A}
	    \left[(\id_X\otimes\Phi_{t,a,1}\otimes\id_B)\big(\tau^{(t,a,1)}_{XVWB}\big)+(\id_X\otimes\Phi_{t,a,-1}\otimes\id_B)\big(\tau^{(t,a,-1)}_{XVWB}\big)\right].\]
    \end{itemize}
    Notice that the evolution step might introduce dependencies between $X,V$ and $B$. However, such dependencies are eliminated later due to how we handle the truncation steps \eqref{eq:trunc}, and thus do not affect our proof.
    
    We can check that the combined partial system $\taut_{XVWB}$ defined above is consistent with the partial system $\taut_{XVW}$ that we discussed in previous sections, in the sense that $\Tr\mt_B[\taut_{XVWB}]=\taut_{XVW}$ always holds:
    \begin{itemize}
        \item For the truncation step, it is straightforward to check that
        \[\Tr\mt_B[\tau^{(t,a)}_{XVWB}]=\sum_{w\in\mc W}\tau^{(t,a)}_{XV|w}\otimes\ket{w}\bra{w}=\tau^{(t,a)}_{XVW}.\]
        \item The permutation operator $U_a$ acts on $\mc W$ as identity since
        \[\Tr\mt_B\left[U_a\ket{w,\beta}\bra{w,\beta}U_a\ct\right]=\ket{w}\bra{w}.\]
        Recalling \cref{eq:tauab} that $\tau^{(t,a,b)}_{XVW}=
	    (\Diag \mathbbm{1}_{a,b}\otimes \id_V)\tau^{(t,a)}_{XVW}$,
	    we have $\Tr\mt_B[\tau^{(t,a,b)}_{XVWB}]=\tau^{(t,a,b)}_{XVW}$.
	    \item The evolution step can be checked directly from the formula without $B$ (\cref{eq:evolution}):
	    \[\tau^{(t+1)}_{XVW}=\E_{a\in\mt_R\mc A}
	    \left[(\id_X\otimes\Phi_{t,a,1})\big(\tau^{(t,a,1)}_{XVW}\big)+(\id_X\otimes\Phi_{t,a,-1})\big(\tau^{(t,a,-1)}_{XVW}\big)\right].\]
    \end{itemize}
    So all previously proved properties about $\taut_{XVW}$ are preserved. In addition, we prove the following two properties about badness levels.
	
	\begin{lemma}\label{lemma:ipbound}
		For every $0\leq t\leq T$, $\ket{v}\in\mc V$ and $w\in\mc W$, we have
		\[\langle P^{\taut}_{X|v,w}, P \rangle\leq \sum_{\beta=0}^T P^{\taut}_{B|w}(\beta)\cdot 2^\beta\cdot 2^{-n}\cdot(1-2^{-r})^{-3t}.\]
	\end{lemma}
	\begin{proof}
		We prove it by induction on $t$. For $t=0$ the lemma is true as $\langle P^{\taut}_{X|v,w}, P \rangle=2^{-n}$ and $P^{\taut}_{B|w}(0)=1$.
		
		Suppose the lemma holds true for some $t<T$. By a similar argument as in \cref{lemma:2norm} and applying \cref{lemma:id_dist} multiple times, we know that for every $\ket{v}\in\mc V, w\in\mc W$ and $a\in\mc A$, there exists $\ket{v'}$and $\ket{v''}\in\mc V$ such that
		\[\langle P^{\tau^{(t,a)}}_{X|v,w}, P \rangle
		=\langle P^{\tau^{(t,\circ)}}_{X|v',w}, P \rangle
		\leq (1-2^{-r})^{-1}\langle P^{\tau^{(t,\star)}}_{X|v',w}, P \rangle
		=(1-2^{-r})^{-1}\langle P^{\taut}_{X|v'',w}, P \rangle,\]
		and therefore
		\begin{equation}\label{eq:induction}
		\langle P^{\tau^{(t,a)}}_{X|v,w}, P \rangle\leq \sum_{\beta=0}^T P^{\taut}_{B|w}(\beta)\cdot 2^\beta\cdot 2^{-n}\cdot(1-2^{-r})^{-3t-1}.
		\end{equation}
		Also, the truncation step by $G_a$ implies that $|\langle M_a,  P^{\tau^{(t,a)}}_{X|v,w} \rangle|\leq 2^{-r}$. That is, for both $b\in\{-1,1\}$,
		\[1-2^{-r}\leq 2\lN{\mathbbm{1}_{a,b}\cdot P^{\tau^{(t,a)}}_{X|v,w}}_1\leq 1+2^{-r}.\]
		Therefore we have, unconditionally
		\begin{equation}\label{eq:bad}
		\langle P^{\tau^{(t,a,b)}}_{X|v,w}, P\rangle = 
		\frac{\langle \mathbbm{1}_{a,b}\cdot P^{\tau^{(t,a)}}_{X|v,w}, P \rangle}
		{\lN{\mathbbm{1}_{a,b}\cdot P^{\tau^{(t,a)}}_{X|v,w}}_1}  \leq
		2(1-2^{-r})^{-1}\cdot\langle P^{\tau^{(t,a)}}_{X|v,w}, P \rangle.
		\end{equation}
		
		When the inner product is evenly divided, i.e. $|\langle M_a,  P^{\sigma^{(t,a)}}_{X|v,w} \rangle|\leq 2^{-r}$, we further have
		\[\langle \mathbbm{1}_{a,b}\cdot P^{\tau^{(t,a)}}_{X|v,w}, P \rangle
		\leq \frac{1}{2}(1+2^{-r})\langle P^{\tau^{(t,a)}}_{X|v,w}, P \rangle
		\leq \frac{1}{2}(1-2^{-r})^{-1}\langle P^{\tau^{(t,a)}}_{X|v,w}, P \rangle,\]
		which means that
		\begin{equation}\label{eq:good}
		\langle P^{\tau^{(t,a,b)}}_{X|v,w}, P\rangle = 
		\frac{\langle \mathbbm{1}_{a,b}\cdot P^{\tau^{(t,a)}}_{X|v,w}, P \rangle}
		{\lN{\mathbbm{1}_{a,b}\cdot P^{\tau^{(t,a)}}_{X|v,w}}_1}  \leq
		(1-2^{-r})^{-2}\cdot\langle P^{\tau^{(t,a)}}_{X|v,w}, P \rangle.
		\end{equation}
		Now there are three cases to discuss:
		\begin{itemize}
		    \item If $(w,a)$ is bad, we have $P^{\tau^{(t,a,b)}}_{B|w}(\beta)=P^{\taut}_{B|w}(\beta-1)$ for every $\beta>0$. Notice that $P^{\taut}_{B|w}(T)=0$ as $t<T$, and thus \cref{eq:induction} and \cref{eq:bad} imply that
		    \begin{align*}
		        \langle P^{\tau^{(t,a,b)}}_{X|v,w}, P\rangle & \leq \sum_{\beta=0}^{T-1} P^{\taut}_{B|w}(\beta)\cdot 2^{\beta+1}\cdot 2^{-n}\cdot(1-2^{-r})^{-3t-2} \\
		        & \leq \sum_{\beta=0}^T P^{\tau^{(t,a,b)}}_{B|w}(\beta)\cdot 2^{\beta}\cdot 2^{-n}\cdot(1-2^{-r})^{-3(t+1)}.
		    \end{align*}
		    \item If $(w,a)$ is not bad and $|\langle M_a,  P^{\sigma^{(t,a)}}_{X|v,w} \rangle|\leq 2^{-r}$, we have $P^{\tau^{(t,a,b)}}_{B|w}(\beta)=P^{\taut}_{B|w}(\beta)$ for every $\beta\geq 0$. Then \cref{eq:induction} and \cref{eq:good} imply that
		    \begin{align*}
		        \langle P^{\tau^{(t,a,b)}}_{X|v,w}, P\rangle & \leq \sum_{\beta=0}^T P^{\taut}_{B|w}(\beta)\cdot 2^{\beta}\cdot 2^{-n}\cdot(1-2^{-r})^{-3t-3} \\
		        & = \sum_{\beta=0}^T P^{\tau^{(t,a,b)}}_{B|w}(\beta)\cdot 2^{\beta}\cdot 2^{-n}\cdot(1-2^{-r})^{-3(t+1)}.
		    \end{align*}
		    \item If $(w,a)$ is not bad and $|\langle M_a,  P^{\sigma^{(t,a)}}_{X|v,w} \rangle|>2^{-r}$, by the definition of badness \eqref{eq:badness} we must have $\langle P^{\tau^{(t,a)}}_{X|v,w}, P \rangle<\frac{1}{2}\cdot 2^{-n}$. Thus by \cref{eq:bad},
		    \[\langle P^{\tau^{(t,a,b)}}_{X|v,w}, P\rangle< (1-2^{-r})^{-1}\cdot 2^{-n}
		    \leq \sum_{\beta=0}^T P^{\tau^{(t,a,b)}}_{B|w}(\beta)\cdot 2^{\beta}\cdot 2^{-n}\cdot(1-2^{-r})^{-3(t+1)}.\]
		\end{itemize}
		The last inequality follows from $\sum_{\beta=0}^T P^{\tau^{(t,a,b)}}_{B|w}(\beta)\cdot 2^{\beta}\cdot 2^{-n}\cdot(1-2^{-r})^{-3(t+1)} \geq 2^{-n} (1-2^{-r})^{-3 (t+1)}$. 
		Hence we obtain the same conclusion from all three cases.
		
		For the evolution step, since $B$ is classical we can view $X$ and $B$ as a whole and apply \cref{claim:chan} on $P^{\tau^{(t+1)}}_{XB|v,w}$, which asserts that $P^{\tau^{(t+1)}}_{XB|v,w}$ is a convex combination of $P^{\tau^{(t,a,b)}}_{XB|v',w'}$ for some $a,b,w'$ and $\ket{v'}$. Then by linearity we conclude that 
		\footnote{It should be noted that in $\tau^{(t+1)}$, $X$ and $B$ are not independent. (In $\tau^{(t,a,b)}$ they are independent (conditioned on $v',w'$)). Nevertheless, independence of $X,B$ (in $\tau^{(t+1)}$) is not needed or used here and we can conclude the final inequality by linearity by taking the corresponding convex combination of all inequalities.}
		\[\langle P^{\tau^{(t+1)}}_{X|v,w}, P \rangle\leq \sum_{\beta=0}^T P^{\tau^{(t+1)}}_{B|w}(\beta)\cdot 2^\beta\cdot 2^{-n}\cdot(1-2^{-r})^{-3(t+1)}. \qedhere\]
	\end{proof}

	\begin{lemma}\label{lemma:badnessweight}
		For every $0\leq\beta\leq t\leq T$ we have
		\[\bra{\beta}\taut_B\ket{\beta}\leq 2^{-k\beta}\binom{t}{\beta}.\]
	\end{lemma}
	\begin{proof}
		We prove it by induction on $t$. For $t=0$ the lemma holds as $\tau^{(0)}_B=\ket{0}\bra{0}\mt_B$. Also notice that the lemma is trivially true for every $t$ when $\beta=0$.
		
		Now suppose the lemma holds for some $t$. By definition we have
		\[\tau^{(t+1)}_B=\E_{a\in\mt_R \mc A}[\tau^{(t,a,1)}_B+\tau^{(t,a,-1)}_B]
		=\E_{a\in\mt_R \mc A}\Tr\mt_W[U_a\tau^{(t,a)}_{WB}U_a\ct].\]
		Therefore
		\[\bra{\beta}\tau^{(t+1)}_B\ket{\beta}=\sum_{w\in\mc W}\E_{a\in\mt_R \mc A} \left[ \bra{w,\beta}U_a\tau^{(t,a)}_{WB}U_a\ct\ket{w,\beta}\right].\]
		By \cref{lemma:badness} we know that for every $w\in\mc W$, the probability that $(w,a)$ is bad for $a\in\mt_R \mc A$ is at most $2^{-k}$. In other words, for every $\beta>0$, 
		\[U_a\ct\ket{w,\beta}=\left\{\begin{array}{ll}
		\ket{w,\beta}, & \textrm{ w.p. }\geq 1-2^{-k} \\
		\ket{w,\beta-1}, & \textrm{ w.p. }\leq 2^{-k}\end{array}\right.\]
		where the probability is taken over the random choice of $a$. 
		It means that
		\begin{align*}
		\bra{\beta}\tau^{(t+1)}_B\ket{\beta} &\leq 
		\sum_{w\in\mc W}\bra{w,\beta}\tau^{(t,a)}_{WB}\ket{w,\beta} + 
		2^{-k}\sum_{w\in\mc W}\bra{w,\beta-1}\tau^{(t,a)}_{WB}\ket{w,\beta-1} \\
		&= \bra{\beta}\tau^{(t,a)}_B\ket{\beta}+ 2^{-k}\cdot \bra{\beta-1}\tau^{(t,a)}_B\ket{\beta-1}.
		\end{align*}
		Notice that
		\[\tau^{(t,a)}_B
		    =\sum_{w\in\mc W}\Tr[\tau^{(t,a)}_{XV|w}]\cdot \Diag P^{\taut}_{B|w} 
		    \leq \sum_{w\in\mc W}\Tr[\taut_{XV|w}]\cdot \Diag P^{\taut}_{B|w}
		    =\taut_B,\]%\qipeng{should we add the reason for the inequality?}
		and thus we conclude that
		\begin{align*}
		\bra{\beta}\tau^{(t+1)}_B\ket{\beta} &\leq 
		\bra{\beta}\taut_B\ket{\beta}+ 2^{-k}\cdot \bra{\beta-1}\taut_B\ket{\beta-1} \\
		&\leq 2^{-k\beta}\binom{t}{\beta}+ 2^{-k}\cdot 2^{-k(\beta-1)}\binom{t}{\beta-1}
	    =  2^{-k\beta}\binom{t+1}{\beta}. \qedhere
		\end{align*}
	\end{proof}
	With the lemmas above in hand, we can finally prove \cref{lemma:main}.
	\begin{proof}[Proof for \cref{lemma:main}]
		For the target distribution $P=P^{\taut}_{X|v,w}$ we have $\langle P^{\taut}_{X|v,w}, P \rangle>2^{2\ell}\cdot 2^{-n}$, so by \cref{lemma:ipbound},
		\[\sum_{\beta=0}^T P^{\taut}_{B|w}(\beta)\cdot 2^\beta\cdot (1-2^{-r})^{-3t}>2^{2\ell}.\]
		Since $t\leq T\leq 2^{r-2}$, we have $(1-2^{-r})^{-3t}\leq 2$, and thus
		\[\sum_{\beta=\ell}^T P^{\taut}_{B|w}(\beta)\cdot 2^\beta
		>\frac{1}{2}\left(2^{2\ell}-2\cdot \sum_{\beta=0}^{\ell-1}2^\beta\right)>2^{\ell}.\]
		On the other hand, for every $\beta\geq\ell$, by \cref{lemma:badnessweight}, 
		\[\Tr[\taut_{B|w}]\cdot P^{\taut}_{B|w}(\beta)\leq \bra{\beta}\taut_B\ket{\beta}\leq (2^{-k}t)^\beta<2^{-(k-r)\beta},\]
		and thus by \cref{eq:m},
		\[\Tr[\taut_{X|v,w}]\leq\Tr[\taut_{B|w}]<2^{-\ell}\sum_{\beta=\ell}^T 2^{-(k-r)\beta}\cdot 2^\beta\leq 2\cdot2^{-(k-r)\ell}\leq2^{-2m}\cdot 2^{-4r}. \qedhere\]
	\end{proof}
	
	\section*{Acknowledgement}
	
	We are grateful to Uma Girish for many important discussions and suggestions on the draft of this paper, and to the anonymous reviewers for their helpful comments.
	
	\bibliography{main}

\newcommand{\etalchar}[1]{$^{#1}$}
\begin{thebibliography}{GKK{\etalchar{+}}08}

\bibitem[AAB{\etalchar{+}}19]{arute2019quantum}
Frank Arute, Kunal Arya, Ryan Babbush, Dave Bacon, Joseph~C Bardin, Rami
  Barends, Rupak Biswas, Sergio Boixo, Fernando~GSL Brandao, David~A Buell,
  et~al.
\newblock Quantum supremacy using a programmable superconducting processor.
\newblock {\em Nature}, 574(7779):505--510, 2019.

\bibitem[Aar18]{aaronson2018shadow}
Scott Aaronson.
\newblock Shadow tomography of quantum states.
\newblock In {\em Proceedings of the 50th Annual ACM SIGACT Symposium on Theory
  of Computing}, pages 325--338, 2018.

\bibitem[ACQ22]{aharonov2022quantum}
Dorit Aharonov, Jordan Cotler, and Xiao-Liang Qi.
\newblock Quantum algorithmic measurement.
\newblock {\em Nature communications}, 13(1):1--9, 2022.

\bibitem[ADR02]{aumann2002everlasting}
Yonatan Aumann, Yan~Zong Ding, and Michael~O Rabin.
\newblock Everlasting security in the bounded storage model.
\newblock {\em IEEE Transactions on Information Theory}, 48(6):1668--1680,
  2002.

\bibitem[AR99]{aumann1999information}
Yonatan Aumann and Michael~O Rabin.
\newblock Information theoretically secure communication in the limited storage
  space model.
\newblock In {\em Annual International Cryptology Conference}, pages 65--79.
  Springer, 1999.

\bibitem[Aud07]{audenaert2007sharp}
Koenraad~MR Audenaert.
\newblock A sharp continuity estimate for the von neumann entropy.
\newblock {\em Journal of Physics A: Mathematical and Theoretical},
  40(28):8127, 2007.

\bibitem[BCL20]{bubeck2020entanglement}
Sebastien Bubeck, Sitan Chen, and Jerry Li.
\newblock Entanglement is necessary for optimal quantum property testing.
\newblock In {\em 2020 IEEE 61st Annual Symposium on Foundations of Computer
  Science (FOCS)}, pages 692--703. IEEE, 2020.

\bibitem[BGY18]{BeameGY18}
Paul Beame, Shayan~Oveis Gharan, and Xin Yang.
\newblock Time-space tradeoffs for learning finite functions from random
  evaluations, with applications to polynomials.
\newblock In {\em {COLT}}, volume~75 of {\em Proceedings of Machine Learning
  Research}, pages 843--856. {PMLR}, 2018.

\bibitem[BY21]{broadbent2021device}
Anne Broadbent and Peter Yuen.
\newblock Device-independent oblivious transfer from the
  bounded-quantum-storage-model and computational assumptions.
\newblock {\em arXiv preprint arXiv:2111.08595}, 2021.

\bibitem[CCHL22]{chen2022exponential}
Sitan Chen, Jordan Cotler, Hsin-Yuan Huang, and Jerry Li.
\newblock Exponential separations between learning with and without quantum
  memory.
\newblock In {\em 2021 IEEE 62nd Annual Symposium on Foundations of Computer
  Science (FOCS)}, pages 574--585. IEEE, 2022.

\bibitem[CLO22]{chen2022toward}
Sitan Chen, Jerry Li, and Ryan O’Donnell.
\newblock Toward instance-optimal state certification with incoherent
  measurements.
\newblock In {\em Conference on Learning Theory}, pages 2541--2596. PMLR, 2022.

\bibitem[CM97]{cachin1997unconditional}
Christian Cachin and Ueli Maurer.
\newblock Unconditional security against memory-bounded adversaries.
\newblock In {\em Annual International Cryptology Conference}, pages 292--306.
  Springer, 1997.

\bibitem[CW01]{Carbery2001}
Anthony Carbery and James Wright.
\newblock Distributional and {$L^q$} norm inequalities for polynomials over
  convex bodies in {$\mathbb{R}^n$}.
\newblock {\em Mathematical Research Letters}, 8(3):233--248, 2001.

\bibitem[DFR{\etalchar{+}}07]{damgaard2007tight}
Ivan~B Damg{\aa}rd, Serge Fehr, Renato Renner, Louis Salvail, and Christian
  Schaffner.
\newblock A tight high-order entropic quantum uncertainty relation with
  applications.
\newblock In {\em Annual International Cryptology Conference}, pages 360--378.
  Springer, 2007.

\bibitem[DFSS07]{damgaard2007secure}
Ivan~B Damg{\aa}rd, Serge Fehr, Louis Salvail, and Christian Schaffner.
\newblock Secure identification and qkd in the bounded-quantum-storage model.
\newblock In {\em Annual International Cryptology Conference}, pages 342--359.
  Springer, 2007.

\bibitem[DFSS08]{damgaard2008cryptography}
Ivan~B Damg{\aa}rd, Serge Fehr, Louis Salvail, and Christian Schaffner.
\newblock Cryptography in the bounded-quantum-storage model.
\newblock {\em SIAM Journal on Computing}, 37(6):1865--1890, 2008.

\bibitem[DM02]{dziembowski2002tight}
Stefan Dziembowski and Ueli Maurer.
\newblock Tight security proofs for the bounded-storage model.
\newblock In {\em Proceedings of the thiry-fourth annual ACM symposium on
  Theory of computing}, pages 341--350, 2002.

\bibitem[DM04]{dziembowski2004generating}
Stefan Dziembowski and Ueli Maurer.
\newblock On generating the initial key in the bounded-storage model.
\newblock In {\em International Conference on the Theory and Applications of
  Cryptographic Techniques}, pages 126--137. Springer, 2004.

\bibitem[DQW21]{dodis2021speak}
Yevgeniy Dodis, Willy Quach, and Daniel Wichs.
\newblock Speak much, remember little: Cryptography in the bounded storage
  model, revisited.
\newblock {\em Cryptology ePrint Archive}, 2021.

\bibitem[DQW22]{dodis2022authentication}
Yevgeniy Dodis, Willy Quach, and Daniel Wichs.
\newblock Authentication in the bounded storage model.
\newblock In {\em Annual International Conference on the Theory and
  Applications of Cryptographic Techniques}, pages 737--766. Springer, 2022.

\bibitem[DR02]{ding2002hyper}
Yan~Zong Ding and Michael~O Rabin.
\newblock Hyper-encryption and everlasting security.
\newblock In {\em Annual Symposium on Theoretical Aspects of Computer Science},
  pages 1--26. Springer, 2002.

\bibitem[FvdG99]{fuchs1999cryptographic}
Christopher~A Fuchs and Jeroen van~de Graaf.
\newblock Cryptographic distinguishability measures for quantum-mechanical
  states.
\newblock {\em IEEE Transactions on Information Theory}, 45(4):1216--1227,
  1999.

\bibitem[GKK{\etalchar{+}}08]{Gavinsky08}
Dmitry Gavinsky, Julia Kempe, Iordanis Kerenidis, Ran Raz, and Ronald de~Wolf.
\newblock Exponential separation for one-way quantum communication complexity,
  with applications to cryptography.
\newblock {\em {SIAM} J. Comput.}, 38(5):1695--1708, 2008.

\bibitem[GKLR21]{GargKLR21}
Sumegha Garg, Pravesh~K. Kothari, Pengda Liu, and Ran Raz.
\newblock Memory-sample lower bounds for learning parity with noise.
\newblock In {\em {APPROX-RANDOM}}, volume 207 of {\em LIPIcs}, pages
  60:1--60:19. Schloss Dagstuhl - Leibniz-Zentrum f{\"{u}}r Informatik, 2021.

\bibitem[GKR20]{GargKR20}
Sumegha Garg, Pravesh~K. Kothari, and Ran Raz.
\newblock Time-space tradeoffs for distinguishing distributions and
  applications to security of goldreich's {PRG}.
\newblock In {\em {APPROX-RANDOM}}, volume 176 of {\em LIPIcs}, pages
  21:1--21:18. Schloss Dagstuhl - Leibniz-Zentrum f{\"{u}}r Informatik, 2020.

\bibitem[GRT18]{GRT18}
Sumegha Garg, Ran Raz, and Avishay Tal.
\newblock Extractor-based time-space lower bounds for learning.
\newblock In {\em Proceedings of the 50th Annual ACM SIGACT Symposium on Theory
  of Computing}, STOC 2018, page 990–1002, New York, NY, USA, 2018.
  Association for Computing Machinery.

\bibitem[GZ19]{guan2019simple}
Jiaxin Guan and Mark Zhandary.
\newblock Simple schemes in the bounded storage model.
\newblock In {\em Annual International Conference on the Theory and
  Applications of Cryptographic Techniques}, pages 500--524. Springer, 2019.

\bibitem[HHJ{\etalchar{+}}16]{haah2016sample}
Jeongwan Haah, Aram~W Harrow, Zhengfeng Ji, Xiaodi Wu, and Nengkun Yu.
\newblock Sample-optimal tomography of quantum states.
\newblock In {\em Proceedings of the forty-eighth annual ACM symposium on
  Theory of Computing}, pages 913--925, 2016.

\bibitem[HKP20]{huang2020predicting}
Hsin-Yuan Huang, Richard Kueng, and John Preskill.
\newblock Predicting many properties of a quantum system from very few
  measurements.
\newblock {\em Nature Physics}, 16(10):1050--1057, 2020.

\bibitem[HN06]{harnik2006everlasting}
Danny Harnik and Moni Naor.
\newblock On everlasting security in the hybrid bounded storage model.
\newblock In {\em International Colloquium on Automata, Languages, and
  Programming}, pages 192--203. Springer, 2006.

\bibitem[KK12]{kasher2012two}
Roy Kasher and Julia Kempe.
\newblock Two-source extractors secure against quantum adversaries.
\newblock {\em Theory of Computing}, 8(1):461--486, 2012.

\bibitem[KRT17]{kol2017time}
Gillat Kol, Ran Raz, and Avishay Tal.
\newblock Time-space hardness of learning sparse parities.
\newblock In {\em Proceedings of the 49th Annual ACM SIGACT Symposium on Theory
  of Computing}, pages 1067--1080, 2017.

\bibitem[LM00]{Laurent2000}
Beatrice Laurent and Pascal Massart.
\newblock Adaptive estimation of a quadratic functional by model selection.
\newblock {\em The Annals of Statistics}, 28(5), October 2000.

\bibitem[Lu02]{lu2002hyper}
Chi-Jen Lu.
\newblock Hyper-encryption against space-bounded adversaries from on-line
  strong extractors.
\newblock In {\em Annual International Cryptology Conference}, pages 257--271.
  Springer, 2002.

\bibitem[LV21]{liu2021secure}
Jiahui Liu and Satyanarayana Vusirikala.
\newblock Secure multiparty computation in the bounded storage model.
\newblock In {\em IMA International Conference on Cryptography and Coding},
  pages 289--325. Springer, 2021.

\bibitem[Mau92]{maurer1992conditionally}
Ueli~M Maurer.
\newblock Conditionally-perfect secrecy and a provably-secure randomized
  cipher.
\newblock {\em Journal of Cryptology}, 5(1):53--66, 1992.

\bibitem[MM18]{MoshkovitzM18}
Dana Moshkovitz and Michal Moshkovitz.
\newblock Entropy samplers and strong generic lower bounds for space bounded
  learning.
\newblock In {\em {ITCS}}, volume~94 of {\em LIPIcs}, pages 28:1--28:20.
  Schloss Dagstuhl - Leibniz-Zentrum f{\"{u}}r Informatik, 2018.

\bibitem[MSTS04]{moran2004non}
Tal Moran, Ronen Shaltiel, and Amnon Ta-Shma.
\newblock Non-interactive timestamping in the bounded storage model.
\newblock In {\em Annual International Cryptology Conference}, pages 460--476.
  Springer, 2004.

\bibitem[OP04]{ohya2004quantum}
Masanori Ohya and D{\'e}nes Petz.
\newblock {\em Quantum entropy and its use}.
\newblock Springer Science \& Business Media, 2004.

\bibitem[PMLA13]{pironio2013security}
Stefano Pironio, Ll~Masanes, Anthony Leverrier, and Antonio Ac{\'\i}n.
\newblock Security of device-independent quantum key distribution in the
  bounded-quantum-storage model.
\newblock {\em Physical Review X}, 3(3):031007, 2013.

\bibitem[Raz17]{R17}
Ran Raz.
\newblock A time-space lower bound for a large class of learning problems.
\newblock In {\em 2017 IEEE 58th Annual Symposium on Foundations of Computer
  Science (FOCS)}, pages 732--742, 2017.

\bibitem[Raz18]{R16}
Ran Raz.
\newblock Fast learning requires good memory: A time-space lower bound for
  parity learning.
\newblock {\em J. ACM}, 66(1), dec 2018.

\bibitem[RS08]{rencher2008linear}
Alvin~C Rencher and G~Bruce Schaalje.
\newblock {\em Linear models in statistics}.
\newblock John Wiley \& Sons, 2008.

\bibitem[Sch07]{schaffner2007cryptography}
Christian Schaffner.
\newblock Cryptography in the bounded-quantum-storage model.
\newblock {\em arXiv preprint arXiv:0709.0289}, 2007.

\bibitem[Sha14]{shamir2014fundamental}
Ohad Shamir.
\newblock Fundamental limits of online and distributed algorithms for
  statistical learning and estimation.
\newblock {\em Advances in Neural Information Processing Systems}, 27, 2014.

\bibitem[SSV19]{sharan2019memory}
Vatsal Sharan, Aaron Sidford, and Gregory Valiant.
\newblock Memory-sample tradeoffs for linear regression with small error.
\newblock In {\em Proceedings of the 51st Annual ACM SIGACT Symposium on Theory
  of Computing}, pages 890--901, 2019.

\bibitem[SVW16]{steinhardt2016memory}
Jacob Steinhardt, Gregory Valiant, and Stefan Wager.
\newblock Memory, communication, and statistical queries.
\newblock In {\em Conference on Learning Theory}, pages 1490--1516. PMLR, 2016.

\bibitem[Uhl76]{uhlmann1976transition}
Armin Uhlmann.
\newblock The ``transition probability'' in the state space of a *-algebra.
\newblock {\em Reports on Mathematical Physics}, 9(2):273--279, 1976.

\bibitem[Wri16]{wright2016learn}
John Wright.
\newblock {\em How to learn a quantum state}.
\newblock PhD thesis, Carnegie Mellon University, 2016.

\bibitem[WW08]{wehner2008composable}
Stephanie Wehner and J{\"u}rg Wullschleger.
\newblock Composable security in the bounded-quantum-storage model.
\newblock In {\em International Colloquium on Automata, Languages, and
  Programming}, pages 604--615. Springer, 2008.

\end{thebibliography}
	\bibliographystyle{alpha}

	\appendix
	\section{\texorpdfstring{Bounding Parameters for \Cref{thm:main}}{Bounding Parameters}} \label{sec:parameter}
\begin{proof}
For \Cref{eq:qr}, since $r \leq r'/4$ by the definition of $r$ and $q \leq r - 7$, we have:
\begin{align*}
    q + r + 1 - r' \leq q + r + 1 - 4 r \leq - 2 r - 6 \leq - 2 r. 
\end{align*}
For \Cref{eq:lr}, we use the fact that $\ell = \frac{1}{5} (\ell' - 13 r - 2)$, therefore
\begin{align*}
    2 \ell + 9 r - n =& \frac{2}{5} (\ell' - 13 r - 2) + 9r - n \\
    = &  \frac{2}{5} \ell' + \frac{19}{5} r - \frac{4}{5} - n.
\end{align*}
Solving the inequality $\frac{2}{5} \ell' + \frac{19}{5} r - \frac{4}{5} - n \leq -r$ (together with the fact that $\ell' \leq n$) gives us $r \leq \frac{1}{8} \ell' + \frac{1}{6}$, which follows since $r \leq \frac{1}{26} \ell' + \frac{1}{6}$ by definition.

For \Cref{eq:m}, as $r \leq \frac{1}{26} \ell' + \frac{1}{6}$, $r \leq \frac{1}{2} (k'-1)$ and $k = k'-1$, 
\begin{align*}
    k-r & \geq k - (k'-1)/2 = (k'-1)/2, \\
    \ell & = \frac{1}{5} (\ell'- 13 r - 2) \geq \frac{1}{5} \left(  \ell' - \frac{\ell'}{2} - \frac{13}{6} - 2\right) > \frac{1}{10} \ell' - 1.
\end{align*}
When $\ell'$ is sufficiently large, we have $\frac{1}{10} \ell' - 1 > \frac{1}{11} \ell' + 5$. Thereby, using the fact that $m \leq (k' - 1) \ell' / 44$: 
\begin{align*}
    (k-r) \ell & \geq \frac{1}{2}(k'-1) \left( \frac{1}{11} \ell' + 5 \right) \\
            & = \frac{1}{22} (k' - 1) \ell' + \frac{5}{2} (k'-1) \\
            & \geq 2m + 5 r \\
             & \geq 2m + 4 r + 1. \qedhere
\end{align*}
\end{proof}

\section{\texorpdfstring{Proof for \cref{prop:distance_weight}}{Proof for Generalized FvdG}} \label{sec:FG-variant}
We first state a variant of Fuchs-van de Graaf inequality \cite{fuchs1999cryptographic} on fidelity, defined for two partial density operators $\rho$ and $\sigma$ as 
\[F(\rho, \sigma) = \Tr\left[ \sqrt{\sqrt{\rho} \sigma \sqrt{\rho}} \right].\]
\begin{lemma}\label{lem:variant-FG}
  Let $\rho, \sigma$ be two semi-definite operators. Assume $\Tr[\rho] \geq \Tr[\sigma]$. Then
  \begin{align*}
	\frac{1}{2} \lN{\rho - \sigma}_\Tr \leq \sqrt{ \frac{1}{4} (\Tr[\rho] + \Tr[\sigma])^2 - F(\rho, \sigma)} \leq \sqrt{ \Tr[\rho]^2 - F(\rho, \sigma)}. 
  \end{align*}
\end{lemma}
Notice that when $\Tr[\rho] = \Tr[\sigma] = 1$, the above inequality is the original Fuchs-van de Graaf inequality. 
\begin{proof}
    Let $u$ and $v$ be purifications of $\rho$ and $\sigma$, that is, $u,v\in\mc H\otimes \mc H_A$ with $\rho=\Tr_A[uu\ct]$ and $\sigma=\Tr_A[vv\ct]$, where $\mc H$ is the ambient space of $\rho$ and $\sigma$, and $\mc H_A$ is some finite-dimensional Hilbert space.
    
    Let $U$ be a unitary that diagonalizes $uu\ct-vv\ct$, that is there is a diagonal matrix $\Lambda \in \mathbb{C}^{d \times d}$ such that $uu\ct-vv\ct = U \Lambda U^\dagger$.
    Let $p, q \in \mathbb{R}^d_{\geq 0}$ be the diagonal of $U^\dagger uu\ct U$ and $U^\dagger vv\ct U$ respectively. We have 
    \begin{align*}
        u\ct u =\Tr[U^\dagger uu\ct U]&= \sum_{x \in [d]} p(x), \\
        v\ct v =\Tr[U^\dagger vv\ct U]&= \sum_{x \in [d]} q(x), \\
        \lN{uu\ct-vv\ct}_\Tr =\lN{\Lambda}_\Tr&= \sum_{x \in [d]} |p(x) - q(x)|, \\
        |\langle u,v\rangle| = |\langle U\ct u,U\ct v\rangle| &\leq \sum_{x \in [d]} \sqrt{p(x)q(x)}.
    \end{align*}
    Therefore, by Cauchy-Schwarz inequality, 
    \begin{align*}
        \lN[2]{uu\ct-vv\ct}_\Tr &= \left(\sum_{x \in [d]} |p(x) - q(x)|\right)^2 \\
         &= \left(\sum_{x \in [d]} \left|\sqrt{p(x)} - \sqrt{q(x)}\right|\cdot\left|\sqrt{p(x)} + \sqrt{q(x)}\right|\right)^2  \\
         &\leq \left(\sum_{x \in [d]} \left|\sqrt{p(x)} - \sqrt{q(x)}\right|^2\right) \left(\sum_{x \in [d]} \left|\sqrt{p(x)} + \sqrt{q(x)}\right|^2\right)  \\
         &= \left(\sum_{x \in [d]} p(x)+\sum_{x \in [d]} q(x)\right)^2-4\left(\sum_{x \in [d]} \sqrt{p(x)q(x)}\right)^2 \\
         &\leq \left(u\ct u+v\ct v\right)^2-4|\langle u,v\rangle|^2.
    \end{align*}
    Notice that $\lN{\rho - \sigma}_\Tr \leq \lN{uu\ct-vv\ct}_\Tr$, $\Tr[\rho]=u\ct u$ and $\Tr[\sigma]=v\ct v$. By Uhlmann's theorem \cite{uhlmann1976transition}, we can also choose $u$ and $v$ such that $F(\rho,\sigma)=|\langle u,v\rangle|^2$. Plugging them into the above inequality concludes the proof.
\end{proof}

Now we are ready to prove \cref{prop:distance_weight}
	\begin{proof}[Proof for \cref{prop:distance_weight}]
		By Fuchs-van de Graaf inequality, it suffices to prove the following bound on fidelity:
		\[F(\rho,\Pi\rho\Pi)\geq  \Tr[\Pi\rho]^2.\]
		Let $u$ be a purification of $\rho\mt_{XV}$, that is,
		$\rho=\Tr_A[uu\ct]$ for some Hilbert space $\mc H_A$. Then $\big(\Pi\otimes\id_A\big)u$ is a purification of $\Pi\rho\Pi$. By Uhlmann's theorem we have
		\[
			F(\rho,\Pi\rho\Pi)
			\geq \left|u\ct\big(\Pi\otimes\id_A\big)u\right|^2
			=\Tr\left[\big(\Pi\otimes\id_A\big)uu\ct\right]^2
			=\Tr\left[\Pi\cdot\Tr_A[uu\ct]\right]^2 
		=\Tr[\Pi\rho]^2. \qedhere
		\]
	\end{proof}

    \section{Linear Quantum Memory Lower Bound}\label{appendix:C}
In this appendix, we prove Theorem~\ref{thm:C} that shows 
a simpler proof for a linear quantum-memory lower bound (without classical memory).
While Theorem~\ref{thm:C} is qualitatively weaker than our main results in most cases, as it only gives a lower bound for programs with a linear-size quantum memory but without a (possibly quadratic) classical memory, Theorem~\ref{thm:C} is technically incomparable to the main results, as it's stated in terms of quantum extractors and the bound on the quantum-memory size depends on a different parameter of the extractor.
Additionally, the proof of Theorem~\ref{thm:C} is significantly simpler than the proof of our main theorem.

We first define the quantum extractor property that we need, which is a simplified version of the ones considered in \cite{kasher2012two}. Given a matrix $M:\mc A\times\mc X\rightarrow\{-1,1\}$, consider two independent sources $A$ and $X$ uniformly distributed over $\mc A$ and $\mc X$ respectively. Suppose there is some quantum register $V$ whose state depends on $A$ and $X$, and they together form a classical-quantum system
\[\rho\mt_{AXV}=\bigoplus_{a\in\mc A,x\in\mc X}\rho\mt_{V|a,x},\]
where $\rho\mt_{V|a,x}$ is the state of $V$ when $A=a$ and $X=x$. For any function $f$ on $A\times X$, we say that $V$ depends only on $f(A,X)$ if for any $a,a'\in\mc A$ and $x,x'\in\mc X$, whenever $f(a,x)=f(a',x')$ we have $\rho\mt_{V|a,x}=\rho\mt_{V|a',x'}$. In particular, $V$ depending only on $A$ is equivalent to $V$ being independent of $X$, or $\rho\mt_{XV}=\rho\mt_X\otimes\rho\mt_V$.

We say that $M$ is an $X$-strong $(q,r)$-quantum extractor, if for every classical-quantum system $\rho\mt_{AXV}$, as above, with the $q$-qubit quantum subsystem $V$ that depends only on $A$, it holds that
\[\left\|\rho\mt_{M(A,X)XV}-U\otimes \rho\mt_{X}\otimes\rho\mt_V\right\|_\Tr\leq 2^{-r}.\]
Here $U=\begin{pmatrix}1/2&0\\0&1/2\end{pmatrix}$ is the uniform operator over one bit, and $\rho\mt_{M(A,X)XV}$ is the classical-quantum system constructed by adding a new classical register storing the value of $M(A,X)$ and tracing out~$A$. In other words, 
\[\rho\mt_{M(A,X)XV}=\bigoplus_{b\in\{-1,1\},x\in\mc X}\sum_{\substack{a\in\mc A\\M(a,x)=b}}\rho\mt_{V|a,x}.\]
Notice that if we choose $V$ to be trivial, the above inequality immediately implies that $\left|\E[M(A,X)]\right|\leq 2^{-r}$.

As an example, the results in \cite{kasher2012two} imply that the inner product function 
%(for parity learning) 
on $n$ bits, where $\mc A=\mc X=\mathbb{F}_2^n$ and
\[M(a,x)=(-1)^{a\cdot x},\]
is an $X$-strong $(k,n-k)$-quantum extractor for every $2\leq k\leq n$.

In this section we prove the following theorem:
\begin{theorem}\label{thm:C}
	Let $\mc X,\mc A$ be two finite sets with $n=\log_2|\mc X|$. Let $M:\mc A\times\mc X\rightarrow\{-1,1\}$ be a matrix which is a $X$-strong $(q,r)$-quantum extractor. Let $\rho$ be a branching program for the learning problem corresponding to $M$, described by classical-quantum systems $\rhot_{XV}$, with $q/2$-qubit quantum memory $V$ and length $T$, and without classical memory. Then the success probability of $\rho$ is at most
    \[2^{-n}+8T\sqrt{n+q}\cdot 2^{-r/4}.\]
\end{theorem}

To prove the theorem, we first need to define the following measure of dependency:
\begin{definition}
	Let $\rho\mt_{XV}$ be a classical-quantum system over classical $X$ and quantum $V$. The dependency of $V$ on $X$ in $\rho\mt_{XV}$ is defined as
	\[\xi^\rho(X;V)=\min_{\tau\mt_V} \lN{\rho\mt_{XV}-\rho\mt_X\otimes \tau\mt_V}_\Tr\]
	where $\tau\mt_V$ is taken over all density operators on $V$. Notice that in this definition taking $\tau\mt_V=\rho\mt_V$ is almost optimal as we have
    \begin{equation}\label{eq:tauv}
    \lN{\rho\mt_{XV}-\rho\mt_X\otimes \rho\mt_V}_\Tr
		\leq \lN{\rho\mt_{XV}-\rho\mt_X\otimes \tau\mt_V}_\Tr+\lN{\rho\mt_V-\tau\mt_V}_\Tr
		\leq 2\lN{\rho\mt_{XV}-\rho\mt_X\otimes \tau\mt_V}_\Tr.
    \end{equation}
\end{definition}

We also consider the quantum mutual information between $X$ and $V$ in $\rho\mt_{XV}$, defined as
\[\I[\rho]{X}{V} = \S(\rho\mt_X)+\S(\rho\mt_V)-\S(\rho\mt_{XV})=\S\left(\rho\mt_{XV}\midd \rho\mt_X\otimes\rho\mt_V\right),\]
where $\S(\cdot)$ denotes the von Neumann entropy, and $\S(\cdot\midd\cdot)$ denotes quantum relative entropy.
When $V$ consists of $q$ qubits, we have the following relationship between our dependency measure and quantum mutual information:
\begin{lemma}\label{lemma:C1}
	$\displaystyle \frac{1}{2}\xi^\rho(X;V)^2\leq\I[\rho]{X}{V}\leq q\cdot \xi^\rho(X;V)+2\sqrt{\xi^\rho(X;V)}$.
\end{lemma}
	\begin{proof}
		On one hand, using the inequality on quantum relative entropy and trace distance (see e.g. \cite[Theorem 1.15]{ohya2004quantum}), we have
		\[\I[\rho]{X}{V}=\S\left(\rho\mt_{XV}\midd \rho\mt_X\otimes\rho\mt_V\right)
		\geq \frac{1}{2}\lN[2]{\rho\mt_{XV}-\rho\mt_X\otimes\rho\mt_V}_\Tr
		\geq \frac{1}{2}\xi^\rho(X;V)^2.\]
		On the other hand, Fannes-Audenaert inequality \cite{audenaert2007sharp} tells us that for every $x\in\mc X$, the difference between the  von Neumann entropies of any two states $\rho$ and $\tau$ on $V$ is bounded by
		\[
		    |\S(\rho)-\S(\tau)|
		    \leq q\cdot \frac{1}{2}\lN{\rho-\tau}_\Tr
		    +h\left(\frac{1}{2}\lN{\rho-\tau}_\Tr\right) 
		\]
		where $h(\epsilon)=-\epsilon\log_2\epsilon-(1-\epsilon)\log_2(1-\epsilon)$ is the binary entropy function. Since the state of $V$ conditioned on $X=x$ is $\rho\mt_{V|x}/\Pr[X=x]=2^n\rho\mt_{V|x}$, we have
		\begin{align*}
			\I[\rho]{X}{V} & = \E_{x\sim X}\left[\S(\rho\mt_V)-\S(2^n\rho\mt_{V|x})\right] \\
			& \leq \frac{1}{2}q\cdot \E_{x\sim X}\lN{\rho\mt_V-2^n\rho\mt_{V|x}}_\Tr
			+\E_{x\sim X} h\left(\frac{1}{2}\lN{\rho\mt_V-2^n\rho\mt_{V|x}}_\Tr\right)  \\
			& \leq \frac{1}{2}q\cdot\lN{\rho\mt_{XV}-\rho\mt_X\otimes \rho\mt_V}_\Tr
			+ h\left(\frac{1}{2}\lN{\rho\mt_{XV}-\rho\mt_X\otimes \rho\mt_V}_\Tr\right)\\
			&  \leq \frac{1}{2}q\cdot\lN{\rho\mt_{XV}-\rho\mt_X\otimes \rho\mt_V}_\Tr
			+ \sqrt{2\lN{\rho\mt_{XV}-\rho\mt_X\otimes \rho\mt_V}_\Tr},
		\end{align*}
		as $h$ is concave and $h(\epsilon)\leq 2\sqrt{\epsilon}$. Now let $\tau_V$ be the optimal density operator in the definition of $\xi^\rho(X;V)$. 
		Plugging in \cref{eq:tauv}, we conclude that
		\[\I[\rho]{X}{V}\leq q\cdot \xi^\rho(X;V)+2\sqrt{\xi^\rho(X;V)}. \qedhere\]
	\end{proof}
	
	\begin{lemma}
		For every classical-quantum system $\rho\mt_{AXV}$ with the $q$-qubit quantum subsystem $V$ that depends only on $A$, we have
		\[\I[\rho]{X}{M(A,X),V}\leq 2(n+q)\cdot 2^{-r/2}.\]
	\end{lemma}
	\begin{proof}
		Since $\I[\rho]{X}{V}=0$, it suffices to bound $\I[\rho]{X}{M(A,X)\mid V}\leq \I[\rho]{M(A,X)}{X,V}$. To bound the later, we first notice that since $M$ is a strong $(q,r)$-quantum extractor,
		\begin{align*}
		    \xi^\rho (M(A,X);X,V) 
		    &\leq \left\|\rho\mt_{M(A,X)XV}-\rho\mt_{M(A,X)}\otimes \rho\mt_X\otimes\rho\mt_V\right\|_\Tr \\
		    &\leq \left\|\rho\mt_{M(A,X)XV}-U\otimes \rho\mt_X\otimes\rho\mt_V\right\|_\Tr +
		    \left|\E[M(A,X)]\right| \\
		    &\leq 2\cdot 2^{-r}.
		\end{align*}
		As the total dimension of $X$ and $V$ is $2^{n+q}$, by \cref{lemma:C1} we have
		\begin{align*}
			\I[\rho]{X}{M(A,X),V} &\leq \I[\rho]{M(A,X)}{X,V} \\
            &\leq (n+q)\cdot \xi^\rho (M(A,X);X,V) +2\sqrt{\xi^\rho (M(A,X);X,V) } \\
            &\leq 5(n+q)\cdot 2^{-r/2}. \qedhere
		\end{align*}
	\end{proof}

	\begin{lemma}
		For every classical-quantum system $\rho\mt_{AXV}$ with $q/2$-qubit quantum subsystem $V$ that depends only on $A$ and $M(A,X)$, we have 
		\[\xi^\rho(X;V)\leq 4\sqrt{n+q}\cdot 2^{-r/4}.\]
	\end{lemma}
	\begin{proof}
		Let $W=\rho_{a,0}\otimes\rho_{a,1}$, where $\rho_{a,b}$ is the density matrix of $V$ when $A=a$ and $M(A,X)=b$. Then $W$ is a $q$-bit quantum system that depends only on $A$. Since $V$ can be decided from $M(A,X)$ and $W$, we have
		\[\xi^\rho(X;V)^2\leq2\I[\rho]{X}{V}\leq 2\I[\rho]{X}{M(A,X)),W}\leq 10(n+q)\cdot 2^{-r/2}. \qedhere\]
	\end{proof}
 
 We are now ready to prove \cref{thm:C}. Let $\Phi_{t,a,b}$ be the quantum channel applied on $V$ at stage $t$ with sample $(a,b)$, and recall that the evolution of the system $\rhot_{XV}$ can be expressed as
 \[\rho^{(t+1)}_{XV}=\E_{a\sim A}
	\left[\sum_{x\in\mc X}\ket{x}\bra{x}\otimes\Phi_{t,a,M(a,x)}\big(\rhot_{V|x}\big)\right].\]
 \begin{proof}[Proof of \cref{thm:C}]
     We are going to bound the increment of $\xi_t$, which is the shorthand for $\xi^{\rho^{(t)}}(X;V)$. For now let us focus on some stage $t$, and let $\tau$ be the density operator that minimizes $\xi_t=\lN{\rhot_{XV}-\rhot_X\otimes \tau}_\Tr$. Notice that $\rhot_X=\rho\mt_X=2^{-n}\id_X$ for every $t$.
     
     Since $\tau$ is a fixed quantum state, we can prepare $\tau$ and apply $\Phi_{A,M(A,X)}$ on $\tau$ to obtain a new quantum register $V'$, which depends only on $A$ and $M(A,X)$. Notice that
     \[\rho\mt_{XV'}=\E_{a\sim A}
	\left[\sum_{x\in\mc X}\ket{x}\bra{x}\otimes\Phi_{t,a,M(a,x)}(\tau)\right],\]
     and therefore similarly to \cref{lemma:evol} (that evolution does not increase trace distance), we can show that
     \begin{align*}
         \left\|\rho^{(t+1)}_{XV}-\rho\mt_{XV'}\right\|_\Tr
         &\leq \E_{a\sim A}\sum_{x\in\mc X}\left\|\Phi_{t,a,M(a,x)}\big(\rhot_{V|x}\big)-
         \Phi_{t,a,M(a,x)}(\tau)\right\|_\Tr \\
         &\leq \sum_{x\in\mc X}\lN{\rhot_{V|x}-\tau}_\Tr 
         \leq \left\|\rhot_{XV}-\rho\mt_X\otimes\tau\right\|_\Tr=\xi_t.
     \end{align*}
     Hence we have
     \begin{align*}
         \xi_{t+1}&\leq\left\|\rho^{(t+1)}_{XV}-\rho\mt_X\otimes \rho\mt_{V'}\right\|_\Tr \\
         &\leq \left\|\rho^{(t+1)}_{XV}-\rho\mt_{XV'}\right\|_\Tr+
         \lN{\rho\mt_{XV'}-\rho\mt_X\otimes \rho\mt_{V'}}_\Tr \\
         &\leq \xi_t+2\xi^\rho(X;V') \\
         &\leq \xi_t+8\sqrt{n+q}\cdot 2^{-r/4}.
     \end{align*}
     Since $\xi_0=0$, we conclude that
     \[\xi\mt_T\leq 8T\sqrt{n+q}\cdot 2^{-r/4}.\]
     This value bounds the difference of the success probability of $\rho$, and that of a quantum branching program whose memory is independent of $X$. The later is clearly at most $2^{-n}$, which finishes the proof.
 \end{proof}

\end{document}